\documentclass[journal]{IEEEtran}

\usepackage{amsmath}
\usepackage{mathptmx}
\usepackage{tabularx}
\usepackage{graphicx}
\usepackage{subfigure}
\usepackage{epstopdf}
\usepackage{amssymb}
\usepackage{amsthm}
\usepackage{color}
\usepackage{algorithm,algpseudocode}
\newtheorem{theorem}{Theorem}

\newtheorem{lemma}{Lemma}
\newtheorem{assumption}{Assumption}
\newtheorem{definition}{Definition}

\usepackage{cite}
\begin{document}

\bstctlcite{IEEEexample:BSTcontrol}

\title{Interdependent Strategic Security Risk Management with Bounded Rationality in the Internet of Things }
\author{Juntao Chen~\IEEEmembership{Student Member,~IEEE,} and Quanyan Zhu~\IEEEmembership{Member,~IEEE}
\thanks{This work was supported in part by the National Science Foundation awards SES-1541164, ECCS-1847056, ARO grant W911NF1910041, and a grant through the Critical Infrastructure Resilience Institute (CIRI).}
\thanks{The authors are with the Department of Electrical and Computer Engineering, Tandon School of Engineering, New York University, Brooklyn, NY 11201 USA (E-mail:\{jc6412,qz494\}@nyu.edu).}}

\maketitle

\begin{abstract}
With the increasing connectivity enabled by the Internet of Things (IoT), security becomes a critical concern, and the users should invest to secure their IoT applications. Due to the massive devices in the IoT network, users cannot be aware of the security policies taken by all its connected neighbors. Instead, a user makes security decisions based on the cyber risks he perceives by observing a selected number of nodes.  To this end, we propose a model which incorporates the limited attention or bounded rationality nature of players in the IoT. Specifically, each individual builds a sparse cognitive network of nodes to respond to.  Based on this simplified cognitive network representation, each user then determines his security management policy by minimizing his own real-world security cost. The bounded rational decision-makings of players and their cognitive network formations are interdependent and thus should be addressed in a holistic manner. We establish a games-in-games framework and propose a Gestalt Nash equilibrium (GNE) solution concept to characterize the decisions of agents, and quantify their risk of bounded perception due to the limited attention. In addition, we design a proximal-based iterative algorithm to compute the GNE. With case studies of smart communities, the designed algorithm can successfully identify the critical users whose decisions need to be taken into account by the other users during the security management.
\end{abstract}


\section{Introduction}

Recent years have witnessed a significant growth of urban population. As the growth continues, cities need to become more efficient to serve the surging population.
To achieve this objective, cities need to become smarter with the integration of information and communication techniques (ICTs) and urban infrastructures. Driven by the advances in sensing, computing, storage and cloud technologies, the Internet of Things (IoT) plays a central role in supporting the development of smart city. 
Though IoT enables a highly connected world, the security of IoT becomes a critical concern. There are 5.5 million new things connected every day in 2016,  as we head toward more than 20 billion by 2020 \cite{Gartner}. These IoT devices come from different manufacturers, and they have heterogeneous functionalities and security configurations and policies. No uniform security standards are used for IoT devices as they are developed using different system platforms for various functionalities. Moreover, due to the connections between IoT devices, the security of one device is also dependent on the security of other devices to which it connects. Therefore, the heterogeneity and the interconnectivity of massive heterogeneous IoT have created significant challenges for security management. Fig. \ref{SH} depicts a highly connected smart community enabled by IoT devices. Each household needs to take into account the cyber risks coming from their connected neighbors when securing their devices.

\begin{figure}[!t]
\centering
\includegraphics[width=0.75\columnwidth]{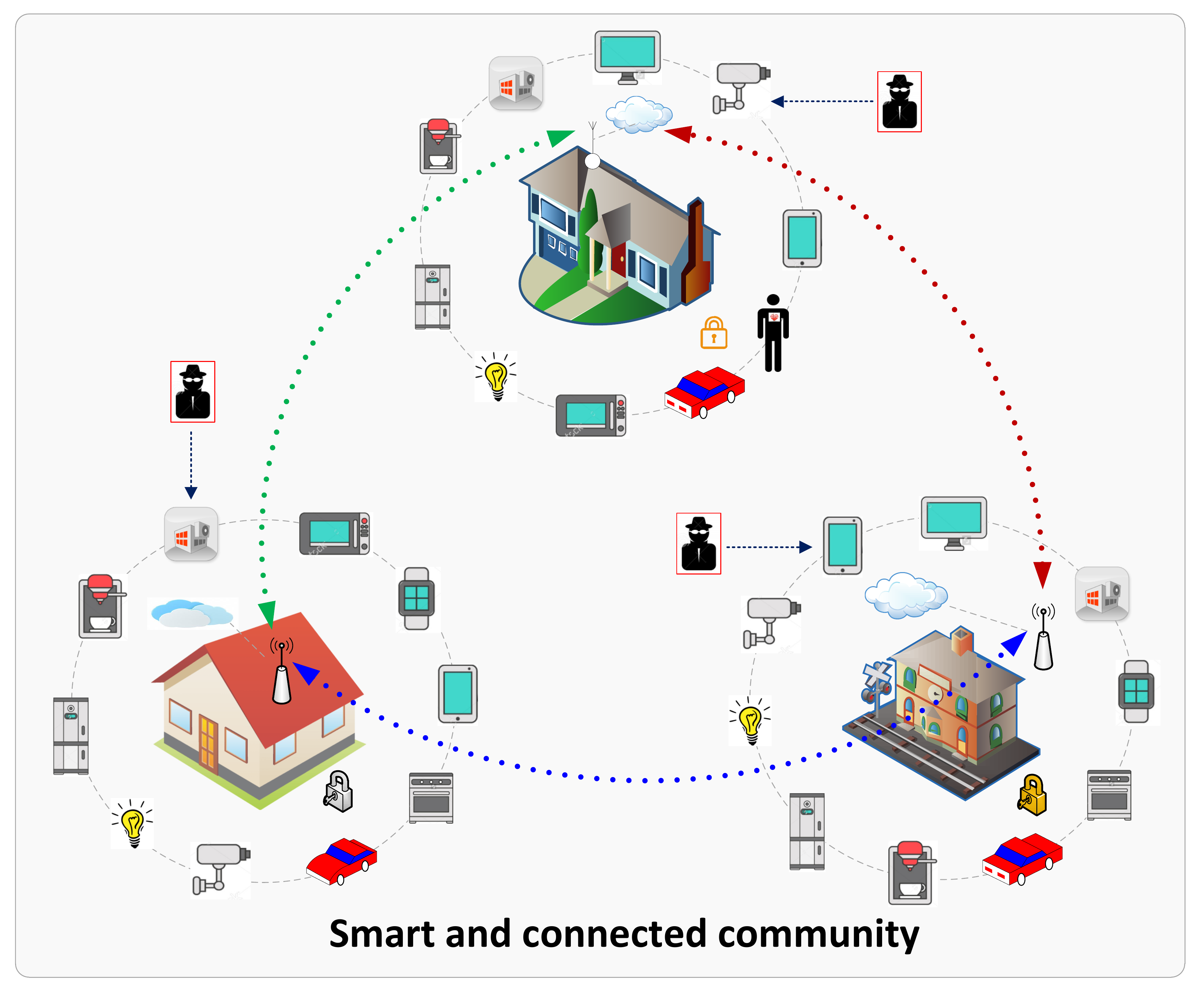}
\caption{IoT-enabled interconnected smart community. The connectivity, on one hand, enhances the situational awareness of smart homes. However, it increses the cyber risks of the community. Hence, the cyber security of each household not only dependents on its own risk management strategy but also the ones of connected neighbors.}\label{SH}
\end{figure}

In cyber networks, security management and practices of users are often viewed as the weakest link \cite{west2008psychology}. The lack of security awareness and expertise at the user's end creates human-induced vulnerabilities that can be easily exploited by an adversary, exacerbating the insecurity of IoT.   To this end, it is critical to enhance the security by strengthening security management in a decentralized way. 
Hence, in the IoT, each device owner or system manager needs to allocate resources (e.g. human resources, computing resources, investments or cognition) to secure his applications. For example, the smart building operator can spend resources on upgrading the hardware, hiring staff members for network monitoring and forensics, and developing tailored security solutions to the smart building. A smart home user, on the other hand, can safely configure its network and regularly updates its software and password of the IoT devices as illustrated in Fig. \ref{SH}.

The devices in the IoT networks and their interconnections can be modeled as nodes and links, respectively. The security policy of one device can have an impact on the security risk of nodes that are connected to it. Since various users own different devices, the security management in IoT is decentralized in nature. Therefore, the process of decentralized security decision-making can be modeled as a game problem in which each user strategically allocates his resources to secure the devices \cite{basar1999dynamic}. In this game, the users' risks are reduced when their connected neighbors are of high-level security.
Due to the complex and massive connections, users cannot be aware of the security policies taken by all its connected neighbors. Instead, a user can only make security decisions based on the cyber risk he perceives by observing a selected number of nodes. This fact indicates that the game model needs to take into account the bounded rationality of players \cite{gabaix2014sparsity}. Therefore, in the game framework, we use a cognition vector representing the observation structure of each IoT user. Specifically, a sparser cognition vector represents a user with weaker cognition ability, and he observes a smaller number of other users' behaviors when deciding his strategy. Thus, the limited attention nature of users creates a bounded perception of cyber risks.

In the established bounded rational game model, the users need to make security management decisions as well as design their cognition networks in a holistic manner. In order to achieve this goal, we define a new solution concept called \textit{Gestalt Nash equilibrium} (GNE) to capture the cognitive network formation and the security management under the bounded rationality simultaneously. The analysis of the GNE provides a quantitative method to understand the risk of massive IoTs and gives tractable security management policies. We further design a proximal-based iterative algorithm to compute the GNE of the game. The GNE resulting from the algorithm reveals several typical phenomena that match well with the real-world observations. For example, when the network contains two groups of users, then under the limited attention, all users will allocate their cognition resources to the same group which demonstrates \textit{the law of partisanship}. Further, in a heterogeneous massive IoT, the equilibrium successfully identifies the set of agents that are invariably paid attention to by other users, demonstrating the phenomenon of \textit{attraction of the mighty}.
Since the framework predicts the high-level systemic risk of the IoT network, it also can be used to inform the design of security standards and incentive mechanisms,  e.g., through contracts and cyber insurance.

The developed security management model provides an essential framework to assess IoT security risks when applied to various applications. For example, in smart home communities, the households are connected together to share heterogeneous information, e.g., electricity prices and temperature readings through smart meters, and real-time information of items in local stores and shops by wireless sensors. The connections of IoT devices thus create security interdependencies between households.
Another broader application lies in the different components in smart cities. Due to the interconnectivity between large-scale infrastructures including the transportation, power grids, and communications, the manager of each sector needs to take into account the cyber risks coming from other components when adopting security solutions.

The contributions of this paper are summarized as follows:
\begin{enumerate}
\item We propose a holistic framework to investigate the security management of users with bounded rationality in the IoT networks.
\item We model the cognition of users with a sparse vector and quantify users' risk of bounded perception resulting from the underperceived cyber threats in the network.
\item We design a proximal-based algorithm to compute the GNE which contains security management strategy and  cognitive network of agents. The algorithm discovers several phenomena including emergence of partisanship, filling the inattention, and attraction of the mighty.
\item We apply the proposed model to a smart community, and demonstrate that the designed algorithm can identify the most critical households in the network.
\end{enumerate}

\subsection{Related Work}
Security management has been investigated in various research fields including computer networks \cite{zhang2017bi}, communications \cite{zhu2012interference},  cloud computing \cite{takabi2010security} and infrastructures \cite{chen2016resilient}. With the advances in ICTs, a growing number of works have focused on the emerging critical issue of IoT security \cite{wu2018game,abie2012risk}.
Due to the interconnectivity between different agents, the security of one agent is also dependent on its connected ones which gives rise to the notion of ``interdependent security'' \cite{kunreuther2003interdependent}. The authors in \cite{xu2015cyber,pawlick2015flip,chen2017interdependent,
pawlick2019istrict} have further investigated the security interdependencies in multilayer cyber-physical systems. The authors in \cite{chen2016optimal,chen2017security,chen2018linear} have developed optimal contracts to address the cyber-physical security issues in IoT. In \cite{chen2017heterogeneous,chen2017dynamic,chen2019}, the authors proposed optimal schemes for designing secure and resilient multi-layer IoT networks through graph-theoretical approaches.

Games over networks have caught a lot of attention recently especially from the economics perspective \cite{jackson2014games, konig2014nestedness,baetz2015social}. The couplings between players in the network can be either in a strategic exclusive or strategic complement manner. Based on the features of security management in IoT, our problem falls into the latter class. For the engineering applications, the authors in  \cite{chen2016resilient,chen2016interdependent} have studied the resource allocation game over interdependent critical infrastructures where both players aim to increase the connectivity of the network. Huang \textit{et.al} \cite{huang2017large,huang2017factored,huang2018distributed,
huang2018factored} have adopted a stochastic Markov game model to design resilient operating strategies for multilayer networks.  Zhu \textit{et.al} \cite{zhu2012guidex} have proposed a game-theoretic framework for collaborative intrusion detection systems through resource management to mitigate network cyber threats. Our work differs from \cite{zhu2012guidex} in that we take into account the cognitive factors of human behaviors during decision making.

Humans with limited knowledge or cognitive resources are bounded rational, since they cannot pay attention to all the information \cite{gigerenzer2002bounded,ellis2018foundations}. Gabaix has proposed a ``sparse max'' operator to model the limited attention of players in which each agent builds a simplified model of the network based on an $l_1$ norm \cite{gabaix2014sparsity}. Built upon \cite{chen2018security} which includes some preliminary results, our work leverages on the established ``sparse max'' operator and formulates a constrained game program to capture the bounded cognition ability of players in the IoT. In addition, we further consider the risk management of each user based on their underperceived cyber risks over the network.

\subsection{Organization of the Paper}
The rest of the paper is organized as follows. Section~\ref{problem} formulates a security management game over IoT networks with bounded rational players. Section \ref{analysis_cog} analyzes the problem. Section~\ref{Al_design} designs a proximal-based iterative algorithm to compute the GNE. Case studies are given in Section~\ref{simulation}, and Section~\ref{conclusion} concludes the paper.

\subsection{Summary of Notations}
For convenience, we summarize the notations used in the paper in Table \ref{table1}. Note that notations associated with $*$ refer to the value at equilibrium. Furthermore, notations with index $k$ stands for its value at step $k$ during the iterative updates.

\begin{table}[t]
\centering
\renewcommand\arraystretch{1.4}
\caption{Nomenclature\label{table1}}
\begin{tabular}{ll} \hline
$\mathcal{N}$ & $\mathcal{N} := \{1,2,...,N\}$, set of players/users\\
$R_{ii}^i$ & security investment cost coefficient of player $i$\\ 
$R_{ij}^i$ & security investment influence coefficient of player $j$ on player $i$\\ 
$r_i$ & unit return of security investment of player $i$\\
$r$ & $r := [r_1,r_2,...,r_N]$\\
$u_{i}$ & security investment decision of player $i$\\
$u$ & $u := [u_1,u_2,...,u_N]$\\
$u_{-i}$ & set of decisions of all players except $i$-th one\\
$\mathcal{U}$ & set of decisions of all players\\
$m^i$ & $m^i:=[m^i_j]_{j\neq i,j\in\mathcal{N}}$, $m^i_j\in[0,1]$, the attention network of player $i$\\
$u_j^{c_i}$ & $u_j^{c_i} = m_j^i u_j$, decision of player $j$ perceived by player $i$ \\
$J_i$ & cost function of player $i$\\ 
$\tilde{J^i}$ & cost function of player $i$ under bounded rationality\\
$BR^{i}$ & best response of player $i$\\
$\Lambda^i$ & $\Lambda^i:=[\Lambda^i_{jk}]_{j\neq i,k\neq i,j\in\mathcal{N},k\in\mathcal{N}}$, $\Lambda^i_{jk} :=  \frac{1}{R_{ii}^i}  R_{ij}^i R_{ik}^i  {u_j}u_k$\\
$e_{N-1}$ & an $N-1$-dimensional column vector with all one entries\\
$\alpha_i$ & weighting factor quantifying the unit cost of player $i$'s cognition\\
$\beta_i$ & total number of links in the cognitive network of player $i$\\
$||\cdot||_1$ & standard L-1 norm\\
$||\cdot||$ & standard L-2 norm\\
$\iota_C$ & indicator function on set $C$\\
$\mathrm{prox}_{\cdot}$ & proximal operator\\
 \hline
\end{tabular}
\end{table}

\section{Problem Formulation}\label{problem}
In this section, we formulate a problem involving strategic security decision making and cognitive network formation of players in the IoT networks.

\subsection{Security Management Game over Networks}
In an IoT user network including a set $\mathcal{N}$ of nodes\footnote{The terms of node, agent and player refer to the user in the IoT, and they are used interchangeably.}, where $\mathcal{N}:=\{1,2,...,N\}$, each node can be seen as a player that makes strategic decisions on the security management to secure their IoT devices. For instance, in Fig. \ref{SH}, each smart home is a player securing their smart things to mitigate the cyber threats.  We define $\mathcal{U}:=\{u_1,...,u_N\}$ by the decision profile of all the players.  Specifically, $u_i$ is a one-dimensional decision variable representing player $i$'s security management effort. For convenience, we denote $u_{-i}:=\mathcal{U}\setminus \{u_i\}$. The objective of player $i$, $i\in\mathcal{N}$, is to minimize his security risk strategically by taking the costly action $u_i$. We define by $F_1^i:\mathbb{R}_+\rightarrow \mathbb{R}_+$ the cost of security management effort of player $i$ which is an increasing function of $u_i$. The corresponding benefit of security management is captured by a function $F_2^i:\mathbb{R}_+\rightarrow \mathbb{R}_+$. Intuitively, a larger $u_i$ yields a higher return, and hence $F_2^i$ is monotonically increasing. Due to the interconnections in the IoT, the risk of player $i$ is also dependent on his connected users. Then, we use a function $F_3^i:\mathbb{R}_+\times \mathbb{R}_+^{N-1}\rightarrow \mathbb{R}_+$ to represent the influence of player $i$'s connected users on his security. The coupling between players in the IoT is in a strategic complement fashion with respect to the security decisions. More specifically, a larger security investment $u_j$ of player $j$, a connected node of player $i$, decreases the cyber risks of player $i$ as well.  Therefore, the cost function of player $i$ can be expressed as the following form:
\begin{equation}\label{obj_general}
J^i(u_i,u_{-i}) = F_1^i(u_i)-F_2^i(u_i)-F_3^i(u_i,u_{-i}),
\end{equation}
where $J^i:\mathbb{R}_+\times \mathbb{R}_+^{N-1}\rightarrow \mathbb{R}$. 
To facilitate the analysis and design of security risk management strategies, we specify some appropriate forms of functions in \eqref{obj_general}. In the following, we focus on player $i$ taking the quadratic form: $F_1^i(u_i) = \frac{1}{2}R_{ii}^i u_i^2$, $F_2^i(u_i)=r_i u_i$, and $F_3^i(u_i,u_{-i}) = \sum_{j\neq i,j\in\mathcal{N}} R_{ij}^i u_i u_j$. Thus, \eqref{obj_general} can be detailed as
\begin{equation}\label{obj}
J^i(u_i,u_{-i}) = \frac{1}{2}R_{ii}^i u_i^2 - r_i u_i -  \sum_{j\neq i,j\in\mathcal{N}} R_{ij}^i u_i u_j,
\end{equation}
where $R_{ii}^i>0,\ r_i>0,\ \forall i$, and $R_{ij}^i\geq0$, $\forall j\neq i,i\in\mathcal{N}$. Note that parameters $R_{ij}^i$, $i,j\in\mathcal{N}$, represent the  risk dependence network of player $i$ in the IoT, and the value of $R_{ij}^i$ indicates the strength of risk influence of player $j$ on player $i$ which is given as a prior. The first term $\frac{1}{2}R_{ii}^i u_i^2$ in \eqref{obj} is the cost of security management with an increasing marginal price.  The second term $r_i u_i$ denotes the corresponding payoff of cyber risk reduction. Then, the first two terms capture the fact that increasing a certain level of cyber security becomes more difficult in a secure network than a less secure one. The last term $\sum_{j=1,j\neq i}^N R_{ij}^i u_i u_j$ is the aggregated security risk effect from connected users of player $i$. Specifically, the structure of $F_3^i$ in $u_i$ and $u_j$ indicates that the risk measure $J^i$ of player $i$ decreases linearly with respect to user $j$'s action. Hence, in the established model, larger investment from a user helps reduce cyber risk influence in a linear way.  We have following assumption on the security influence parameters.
\begin{assumption}\label{assumption1}
$R_{ii}^i>\sum_{j\neq i,j\in\mathcal{N}}R_{ij}^i,\ \forall i\in\mathcal{N}$.
\end{assumption}
Assumption \ref{assumption1} has a natural interpretation which indicates that the security of a user is mainly determined by his own strategy rather than other users' decisions in the IoT network. Moreover, based on the heterogeneous influence networks characterized by Assumption \ref{assumption1}, each node designs its own security investment strategy which enables the decentralized decision-making. The strategies of nodes are interdependent due to the coupling between their cost functions shown in \eqref{obj}.

Through the first order optimality condition (FOC), we obtain
\begin{equation}\label{response}
R_{ii}^i u_i - \sum_{j\neq i,j\in\mathcal{N}}R_{ij}^i u_j-r_i = 0,\ \forall i\in\mathcal{N}.
\end{equation}

Putting \eqref{response} in a matrix form yields
\begin{equation}
\begin{bmatrix}
R_{11}^1 & -R_{12}^1 & \dotsm & -R_{1N}^1\\
-R_{21}^2 & R_{22}^2 & \dotsm & -R_{2N}^2\\
\vdots & \vdots & \ddots & \vdots \\
-R_{N1}^N & -R_{N2}^N & \dotsm & R_{NN}^N
\end{bmatrix}
\begin{bmatrix}
u_1\\
u_2\\
\vdots\\
u_N
\end{bmatrix} = 
\begin{bmatrix}
r_1\\
r_2\\
\vdots\\
r_N
\end{bmatrix}\ \Leftrightarrow \ Ru=r,
\end{equation}
where $r:=[r_i]_{i\in\mathcal{N}}$, $u:=[u_i]_{i\in\mathcal{N}}$.

For convenience, we denote this security management game by $\mathcal{G}$. One solution concept of game $\mathcal{G}$ is Nash equilibrium (NE) which is defined as follows.

\begin{definition}[Nash Equilibrium of Game $\mathcal{G}$ \cite{basar1999dynamic}]
The strategy profile $u^*=[u_i^*]_{i\in\mathcal{N}}$ constitutes a Nash equilibrium of game $\mathcal{G}$ if 
$
J^i(u_i,u_{-i}^*)\geq J^i(u_i^*,u_{-i}^*),\ \forall i\in\mathcal{N},\ \forall u_i\in\mathcal{U}_i.
$
\end{definition}

The NE of game $\mathcal{G}$ yields strategic security management policies of players under the condition that they can perceive all the cyber risks in the IoT network. 

\subsection{Bounded Rational Security Management Game}\label{security_bounded_rational}
In reality, the users in IoT are connected with numerous other agents. For example, a single household can be connected with a number of other houses in terms of various types of IoT products in the smart communities. Therefore, when making security management strategies, each user may not be capable to observe all its connected neighbors. Instead, a user can only respond to a selected number of other players' decisions. Then, this bounded rational response mechanism creates a cognitive network formation process for the players in the network. Specifically, player $i$'s irrationality is captured by a vector $m^i:=[m^i_j]_{j\neq i,j\in\mathcal{N}}$, $m^i_j\in[0,1]$, which stands for the attention network that player $i$ builds. When $m^i_j=0$, user $i$ pays no attention to user $j$'s behavior; when $m^i_j=1$, user $i$ observes the true value of security management $u_j$ of user $j$. The value that $m^i_j$ admits between 0 and 1 can be interpreted as the trustfulness of user $i$ on the perceived $u_j$. Another interpretation of $m^i_j$ can be the probability that user $i$ observes the behavior of user $j$ at each time instance on the security investment over a long period.  Hence, the decision of player $j$ perceived by player $i$ becomes $u_j^{c_i} = m_j^i u_j$. Then, player $i$ minimizes the modified cost function with bounded rationality defined as:
\begin{align}
\tilde{J^i}(u_i,u_{-i}^{c_i},m^i) &= \frac{1}{2} R_{ii}^i u_i^2 - r_i u_i - \sum_{j\neq i,j\in\mathcal{N}}  m_j^i R_{ij}^i u_i  u_j\notag\\
& = \frac{1}{2} R_{ii}^i u_i^2 - r_i u_i - \sum_{j\neq i,j\in\mathcal{N}}  R_{ij}^i u_i u_j^{c_i},\label{obj_with_cog}
\end{align}
where $\tilde{J^i}:\mathbb{R}_+\times \mathbb{R}_+^{N-1}\times [0,1]^{N-1}\rightarrow \mathbb{R}$. 

The FOC of \eqref{obj_with_cog} gives
$
R_{ii}^i u_i - \sum_{j\neq i,j\in\mathcal{N}}R_{ij}^i u_j^{c_i}-r_i = 0,\ \forall i\in\mathcal{N},
$
which is equivalent to 
\begin{equation}\label{response2_matrix}
\begin{split}
\begin{bmatrix}
R_{11}^1 & -m^1_2 R_{12}^1 & \dotsm & -m^1_N R_{1N}^1\\
-m^2_1 R_{21}^2 & R_{22}^2 & \dotsm & -m^2_N R_{2N}^2\\
\vdots & \vdots & \ddots & \vdots \\
-m^N_1 R_{N1}^N & -m^N_2 R_{N2}^N & \dotsm & R_{NN}^N
\end{bmatrix}
\begin{bmatrix}
u_1\\
u_2\\
\vdots\\
u_N
\end{bmatrix} &= 
\begin{bmatrix}
r_1\\
r_2\\
\vdots\\
r_N
\end{bmatrix}\\
 \Leftrightarrow \quad R^s u&=r.
 \end{split}
\end{equation}

The bounded rational best-response of player $i$, $i\in\mathcal{N}$, then becomes
\begin{equation}\label{sparse_BR}
u_i = BR^i(u_{-i}^{c_i}) = \frac{1}{R_{ii}^i}\left(\sum_{j\neq i,j\in\mathcal{N}}R_{ij}^i u_j^{c_i}+r_i\right),
\end{equation} 
where $u_j^{c_i} = m_j^i u_j$.

We denote the security management game of players with limited attention by $\tilde{\mathcal{G}}$. Comparing with the solution concept NE of game $\mathcal{G}$, the one of game $\tilde{\mathcal{G}}$ is generalized to bounded rational Nash equilibrium (BRNE). The formal definition of BRNE is as follows.

\begin{definition}[Bounded Rational Nash Equilibrium of Game $\tilde{\mathcal{G}}$]
With given cognition vectors $m^i$, $\forall i\in\mathcal{N}$, the strategy profile $u^*=[u_i^*]_{i\in\mathcal{N}}$ constitutes a BRNE of game $\tilde{\mathcal{G}}$ if 
$
\tilde{J}^i(u_i,u_{-i}^*,m^i)\geq \tilde{J}^i(u_i^*,u_{-i}^*,m^i),\ \forall i\in\mathcal{N},\ \forall u_i\in\mathcal{U}_i.
$
\end{definition}
Note that the cognitive network each user built has an impact on the BRNE of game $\tilde{\mathcal{G}}$. Hence, how the users determine the cognition vector $m^i$, $i\in\mathcal{N}$, becomes a critical issue. In the ensuing section, we introduce the cognitive network formation of players in the IoT.

\subsection{Cognitive Network Formation}\label{cog_formation}
Due to the massive connections in IoT, each user builds a sparse cognitive network containing the agents to observe. To this end, the real cost of user $i$ by taking the bounded rationality into account becomes
\begin{align*}
J^i&(BR^i(u_{-i}^{c_i}),u_{-i}) 
= \frac{1}{2R_{ii}^i} \left(\sum_{j\neq i,j\in\mathcal{N}}R_{ij}^i u_j^{c_i}+r_i\right)^2 \\
&- \sum_{k\neq i,k\in\mathcal{N}} \left[ \frac{1}{R_{ii}^i} R_{ik}^i u_k \left(\sum_{j\neq i,j\in\mathcal{N}}R_{ij}^i u_j^{c_i}+r_i\right) \right] \\
&-  \frac{r_i}{R_{ii}^i} \left(\sum_{j\neq i,j\in\mathcal{N}}R_{ij}^i u_j^{c_i}+r_i\right)\\
 = &\frac{1}{2} \sum_{j\neq i,j\in\mathcal{N}}\sum_{k\neq i,k\in\mathcal{N}}   \frac{1}{R_{ii}^i} R_{ij}^i  R_{ik}^i {u_j^{c_i}} u_k^{c_i} - \frac{1}{2R_{ii}^i} \left(r_i\right)^2  \\
&- \sum_{k\neq i,k\in\mathcal{N}} \left(\sum_{j\neq i,j\in\mathcal{N}} {u_j^{c_i}} R_{ij}^i \right) \frac{1}{R_{ii}^i} R_{ik}^i u_k - \sum_{k\neq i,k\in\mathcal{N}}  \frac{1}{R_{ii}^i} {r_i} R_{ik}^i u_k.
\end{align*}

Incorporating the cognition vector $m^i$ into the real cost of player $i$ further yields
\begin{equation}\label{real_cost}
\begin{split}
J^i&(BR^i(u_{-i}^{c_i}),u_{-i}) = \\
&\frac{1}{2} \sum_{j\neq i,j\in\mathcal{N}}\sum_{k\neq i,k\in\mathcal{N}}  m_j^i \frac{1}{R_{ii}^i} R_{ij}^i  R_{ik}^i m_k^i {u_j} u_k - \frac{1}{2R_{ii}^i} \left(r_i\right)^2  \\
&- \sum_{k\neq i,k\in\mathcal{N}} \sum_{j\neq i,j\in\mathcal{N}} m_j^i \frac{1}{R_{ii}^i} R_{ij}^i  R_{ik}^i {u_j} u_k - \sum_{k\neq i,k\in\mathcal{N}} \frac{1}{R_{ii}^i} {r_i}  R_{ik}^i u_k.
\end{split}
\end{equation}

Recall that each user aims to minimize the security risk based on the risks he perceives. Thus, by considering the real cost induced by the bounded rationality constraint, the strategic cognitive network formation problem of player $i$ can be formulated as
\begin{align*}
m^{i*} =  & \arg \min_{m^i_j,{j\neq i,j\in\mathcal{N}}} J^i(BR^i(u_{-i}^{c_i}),u_{-i}) + \alpha_i \Vert m^i \Vert_1\\
=&\arg \min_{m^i_j,{j\neq i,j\in\mathcal{N}}}\frac{1}{2} \sum_{j\neq i,j\in\mathcal{N}}\sum_{k\neq i,k\in\mathcal{N}}   \frac{1}{R_{ii}^i} R_{ij}^i  R_{ik}^i  {u_j}u_k m_j^i m_k^i\\
&- \sum_{j\neq i,j\in\mathcal{N}} \sum_{k\neq i,k\in \mathcal{N}}    \frac{1}{R_{ii}^i} R_{ij}^i  R_{ik}^i u_k {u_j} m_j^i + \alpha_i \Vert m^i \Vert_1\\
 = &\arg \min_{m^i_j,{j\neq i,j\in\mathcal{N}}}
\frac{1}{2} {m^i}^T \Lambda^i m^i - e_{N-1}^T\Lambda^i m^i + \alpha_i \Vert m^i \Vert_1,
\end{align*}
where $\Lambda^i:=[\Lambda^i_{jk}]_{j\neq i,k\neq i,j\in\mathcal{N},k\in\mathcal{N}}$, $\Lambda^i_{jk} =  \frac{1}{R_{ii}^i}  R_{ij}^i R_{ik}^i  {u_j}u_k$, $e_{N-1}$ is an $N-1$-dimensional column vector with all one entries, and $\alpha_i$ is a weighting factor capturing the unit cost of cognition of player $i$ and it can be tuned to match with experimental data. The term $\Vert m^i \Vert_1$ is a convex relaxed version of $\Vert m^i \Vert_0$ which approximately maintains the sparse property of player $i$'s cognitive network \cite{candes2006near,baraniuk2007compressive}. The integrated term $\alpha_i \Vert m^i \Vert_1$ can be interpreted as the \textit{cognitive cost} of user $i$.

Therefore, for player $i$, we need to solve the following constrained optimization problem:
\begin{equation}\label{P1}
\begin{split}
\min_{m^i_j,{j\neq i,j\in\mathcal{N}}} &\frac{1}{2} {m^i}^T \Lambda^i m^i - e_{N-1}^T\Lambda^i m^i + \alpha_i \Vert m^i \Vert_1\\
\mathrm{s.t.}\quad & 0\leq m^i_j\leq 1,j\neq i,j\in\mathcal{N}, \ \mathrm{(Risk\  perception)},
\end{split}
\end{equation}
where the constraints $m^i_j\in[0,1]$, $\forall j\neq i$, indicate the risk perception behavior of user $i$.

The number of cognitive links that player $i$ can form is generally a positive integer, i.e., $\Vert m^i \Vert_1 = \beta_i\in\mathbb{N}^+$.  Note that $\beta_i$ here and $\alpha_i$ in \eqref{P1} have the same interpretation which both quantify the cognition ability of player $i$. Then, by choosing $\alpha_i$ strategically, the problem in \eqref{P1} is equivalent to the following problem:
\begin{equation}\label{P_equi}
\begin{split}
\min_{m^i_j,{j\neq i,j\in\mathcal{N}}} &\frac{1}{2} {m^i}^T \Lambda^i m^i - e_{N-1}^T\Lambda^i m^i \\
\mathrm{s.t.}\quad & 0\leq m^i_j\leq 1,j\neq i,j\in\mathcal{N},\ \mathrm{(Risk\  perception)},\\
& \Vert m^i \Vert_1 = \beta_i,\ \mathrm{(Limited\  attention)},
\end{split}
\end{equation}
where $\beta_i\in\mathbb{N}^+\leq N-1$ is the total number of links that player $i$ can form in his cognitive network, quantifying his limited attention. Simulation studies in Section \ref{simulation} reflect that considering $\Vert m^i \Vert_1 = \beta_i$ yields sparser cognitive networks. Note that we still solve \eqref{P1} by selecting a proper $\alpha_i$ which yields equivalent \eqref{P1} and \eqref{P_equi}.

\subsection{Gestalt Nash Equilibrium}
The formulated security management under bounded rationality problem boasts a games-of-games structure. The users make decisions strategically in the IoT network as well as form their cognitive networks selfishly. The security management game and cognitive network formation game are interdependent. Therefore, the cognitive and IoT user layers shown in Fig. \ref{two_networks} constitute a network-of-networks framework. In this paper, we aim to design an integrated algorithm to design the cognitive networks and determine the security risk management decisions of users in a holistic manner.

 \begin{figure}[!t]
\centering
\includegraphics[width=0.8\columnwidth]{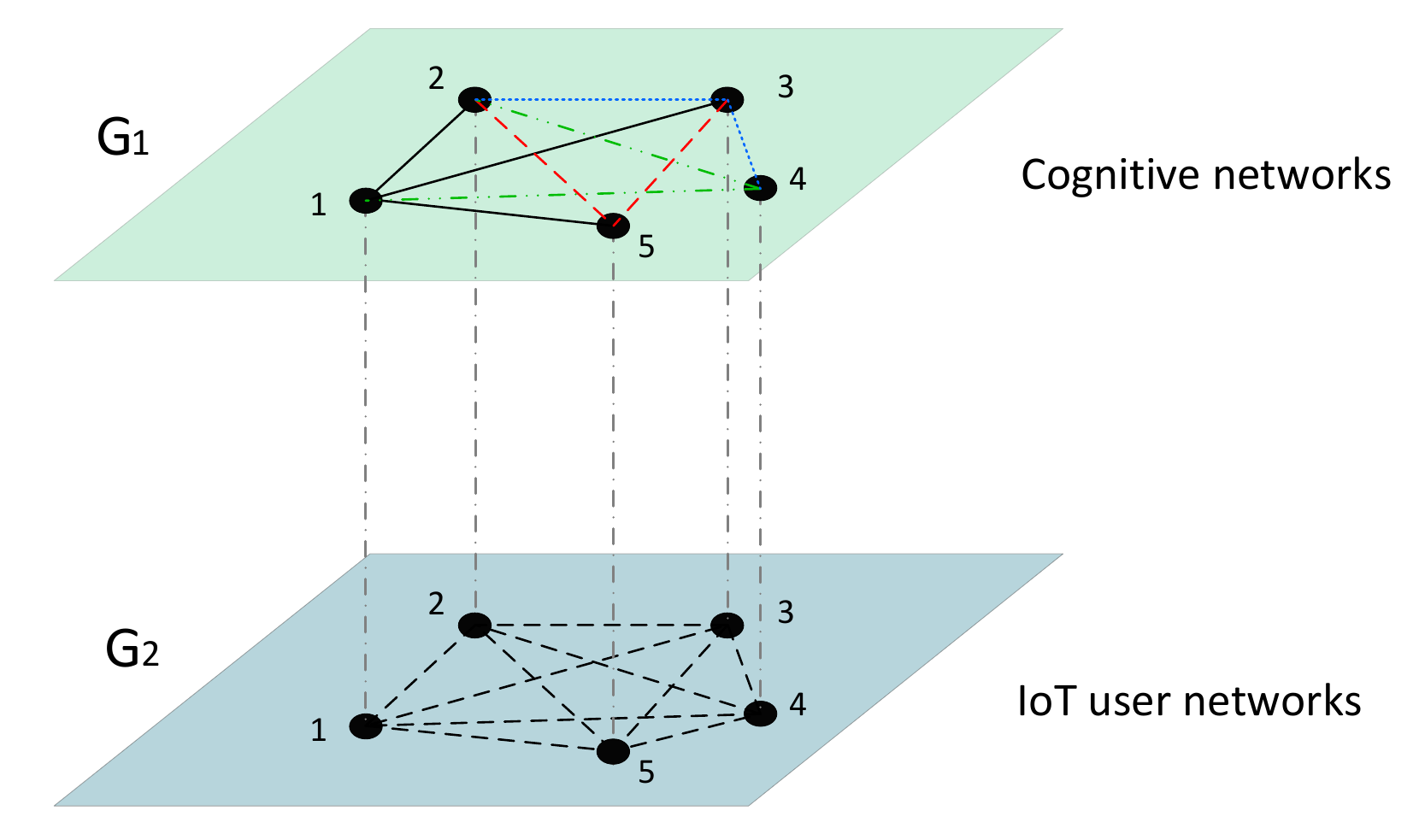}
\caption{IoT user and cognitive network-of-networks. Users make strategic security management decisions in the IoT network as well as determine their cognitive networks. The security management game in layer $G_2$ and the cognitive network formation game in layer $G_1$ are interdependent which create a games-of-games framework.}\label{two_networks}
\end{figure}

To this end, we present the solution concept, Gestalt Nash equilibrium, of the bounded rational security risk management game as follows. 

\begin{definition}[Gestalt Nash Equilibrium]
The Gestalt Nash equilibrium (GNE) of the security risk management game under bounded rationality is a profile $(m^{i*},u_i^*)$, $\forall i\in\mathcal{N}$, that satisfies
\begin{equation*}
\tilde{J}^i(u_i^*,u_{-i}^*,m^{i*})\leq \tilde{J}^i(u_i,u_{-i}^*,m^{i}),\ \forall u_i\in\mathcal{U}_i,\ \forall m^i\in [0,1]^{N-1}.
\end{equation*}
At the GNE, all the players in the network do not change their action $u_i$ and cognition vector $m^i$, $\forall i\in\mathcal{N}$, simultaneously.
\end{definition}

\textit{Remark:} The strategic security management profile $u^*=[u_i^*]_{i\in\mathcal{N}}$ at GNE is also a BRNE. 

In the following, we aim to analyze the GNE of the game and compute it by designing algorithms.

\section{Problem Analysis}\label{analysis_cog}
We first analyze the convergence of the bounded rational best-response dynamics of players in Section \ref{security_bounded_rational}. Then, we quantify the risk of bounded perception due to limited attention of players. We further reformulate the cognitive network formation problem presented in Section \ref{cog_formation}.

\subsection{Bounded Rational Best Response Dynamics}\label{BS_algorithm}
Based on Section \ref{security_bounded_rational}, the bounded rational best-response dynamics of player $i$ under cognitive network $m^{i}$, $i\in\mathcal{N}$, can be written as
\begin{equation}\label{sparse_BR_dynamic}
u_{i,k+1} = BR^i(u_{-i,k}^{c_i}) = \frac{1}{R_{ii}^i}\left(\sum_{j\neq i,j\in\mathcal{N}}R_{ij}^i u_{j,k}^{c_i}+r_i\right),
\end{equation} 
where $u_{j,k}^{c_i} = m_j^i u_{j,k}$ and $k$ denotes the iteration index. Then, 
we obtain the following convergence result of security management strategy updates of users under given cognition networks.

\begin{lemma}
Under Assumption \ref{assumption1}, the sparse best-response dynamics \eqref{sparse_BR_dynamic} for all players converge to a unique BRNE.
\end{lemma}
\begin{proof}
In the sparse cognition networks, $R_{ii}^i>\sum_{j\neq i,j\in\mathcal{N}}m^i_j R_{ij}^i,\ \forall i\in\mathcal{N}$, since $m^i_j\in[0,1]$. Then, $R^s$ defined in \eqref{response2_matrix} is strictly diagonal dominant by rows, and $u$ admits a unique solution. In addition, both Gauss-Seidel and Jacobi types of best-response dynamics \eqref{sparse_BR_dynamic} converges \cite{saad2003iterative}.
\end{proof}

Note that Assumption \ref{assumption1} is a sufficient condition. In some cases, the best-response dynamics \eqref{sparse_BR_dynamic} may still converge when Assumption \ref{assumption1} does not hold. We focus on the scenarios under Assumption \ref{assumption1} in this paper which exhibit a natural security dependence interpretation.

\subsection{Risk of Bounded Perception}
When making security strategies in the IoT,  the risk of bounded perception (RBP) of users due to irrationality/limited attention is defined as follows.
\begin{definition}[RBP]
With the cognition vector $m^i$, the RBP of player $i$, $i\in\mathcal{N}$, is defined as
\begin{equation}
L_i(m^i,u_{-i})=J^i(BR^i(u_{-i}^{c_i}),u_{-i}) - J^i(BR^i(u_{-i}),u_{-i}),
\end{equation}
where $L_i:\mathcal{M}_i \times \mathcal{U}_{-i}\rightarrow \mathbb{R}$.
\end{definition}

Note that RBP is defined over the real-world cost functions \eqref{obj}, quantifying the security loss of the users due to limited attention. We further present the following lemma.
\begin{lemma}\label{lemma1}
Under the bounded rational model, each user in the network has a degraded security level comparing with the one obtained from the model containing fully rational users. The RBP of player $i$, $i\in\mathcal{N}$, with bounded rationality is 
\begin{equation*}
L_i(m^i,u_{-i}) = \frac{1}{2} \sum_{j\neq i,j\in\mathcal{N}} \sum_{k\neq i,k\in\mathcal{N}}  (1-m^i_j)(1-m^i_k) \frac{1}{R_{ii}^i} R_{ji}^i  R_{ik}^i {u_j} u_k.
\end{equation*}
\end{lemma}
\begin{proof}
See Appendix \ref{lemma1_apx}.
\end{proof}

\textit{Remark:} Note that the RBP of each player is nonnegative from Lemma \ref{lemma1}, since the coefficients and security investments are nonnegative and the cognition variable admits a value between 0 and 1. Intuitively, if player $i$ is able to perceive all the cyber risks in the network, i.e., $m^i_j=1$, $\forall j\neq i,\ j\in\mathcal{N}$, then the RBP is $L_i(m^i,u_{-i})=0$. In this scenario, the bounded rational model degenerates to the fully rational one. This indicates that, with more observations, the IoT users can design security management strategies better to lower their security risks. This fact also leads to the conclusion that more information (better cognitive ability) is beneficial for the users in our security management game. The result in Lemma \ref{lemma1} is further illustrated and corroborated through case studies in Section \ref{simulation}.
 
\subsection{Problem Reformulation}
We can rewrite the constrained optimization program \eqref{P1} as
\begin{equation}\label{P2}
\min_{m^i_j,{j\neq i,j\in\mathcal{N}}} Q_i(m^i):= \frac{1}{2} {m^i}^T \Lambda^i m^i - e_{N-1}^T\Lambda^i m^i + \alpha_i \Vert m^i \Vert_1 + \iota_C(m^i),
\end{equation}
where $Q_i:[0,1]^{N-1}\rightarrow \mathbb{R}\cup \{+\infty\}$, $C:=\{m^i|0\leq m^i_j\leq 1,j\neq i,j\in\mathcal{N}\}$, and $\iota_C$ is an indicator function, i.e., 
\begin{equation}
\iota_C(x) = \begin{cases} 0, & \mathrm{if}\ x\in C,\\ +\infty, & \mathrm{otherwise}.\end{cases}
\end{equation}
For convenience, we decompose the function $Q_i$ into three parts and define
\begin{equation}\label{fs}
\begin{split}
f_1^i(m^i) &= \frac{1}{2} {m^i}^T \Lambda^i m^i - e_{N-1}^T\Lambda^i m^i,\ \mathrm{(Security\ loss)},\\
 f_2^i(m^i) &= \alpha_i \Vert m^i \Vert_1,\ \mathrm{(Cognition\ cost)},\\
 f_3^i(m^i) &= \iota_C(m^i),\ \mathrm{(Feasible\ risk\ perception)},
 \end{split}
\end{equation}
where $f_1^i:\mathbb{R}^{N-1}\rightarrow \mathbb{R},$ $ f_2^i:\mathbb{R}^{N-1}\rightarrow [0,+\infty)$ and $f_3^i:\mathbb{R}^{N-1}\rightarrow \{0,+\infty\}$. Specifically, for user $i\in\mathcal{N}$, $f_1^i$ quantifies a modified security loss; $f_2^i$ captures the cognition cost; and $f_3^i$ ensures a feasible risk perception over the IoT.

The optimization problem \eqref{P2} is quite challenging to solve. First, note that the convexity of $f_1^i$ depends on the characteristics of matrix $\Lambda^i$. Specially, when $\Lambda^i$ is positive definite, then $f_1^i$ is convex in $m^i$. When $\Lambda^i$ is not definite, then solving the quadratic program is an NP hard problem. Second, the $l_1$ norm-based function $f_2^i$ and the indicator function $f_3^i$ are nonsmooth and not differentiable, though they are convex. The traditional gradient-based optimization tools are not sufficient to deal with this type of optimization problem in \eqref{P2} \cite{nocedal2006numerical}. To this end, we aim to design a proximal algorithm to solve this problem. 

\section{Computing GNE via Algorithm Design}\label{Al_design}
In this section, our goal is to design an algorithm to solve problem \eqref{P2}. We further characterize the closed form solutions for a special case with homogeneous agents for comparison during case studies in Section \ref{simulation}.  In addition, we present an integrated algorithm that computes the GNE of the bounded rational security management game.

\subsection{Basics of Proximal Operator}
To address \eqref{P2}, we leverage the tools from proximal operator theory. We first present the definition of proximal operator as follows.

\begin{definition}[Proximal Operator \cite{bauschke2011convex}]\label{proximal_operator}
Let $g\in\Gamma_0$, where $\Gamma_0$ denotes the set of proper lower semicontinuous convex functions. The proximal mapping associated to $g$ is defined as
\begin{equation}\label{prox}
\mathrm{prox}_{\lambda g}(x) = \arg\min_{l}\ g(l)+\frac{1}{2\lambda}\Vert l-x \Vert^2. 
\end{equation}
\end{definition}
Note that the proximal mapping is unique, since the optimization problem in \eqref{prox} is convex. Specifically, for function $f_2^i$ in \eqref{fs}, we have 
\begin{align*}
\left[\mathrm{prox}_{\lambda f_2^i}(x)\right]_j = 
\begin{cases}
x_j-\lambda \alpha_i, & x_j\geq \lambda \alpha_i,\\
0, & |x_j|<\lambda \alpha_i,\\
x_j+\lambda \alpha_i, & x_j\leq -\lambda \alpha_i,
\end{cases}
\end{align*}
for $j\neq i,\ j\in\mathcal{N}$, which can be put in a compact form as \cite{parikh2014proximal}
\begin{equation}\label{prox_f2}
\mathrm{prox}_{\lambda f_2^i}(x) = (x-\lambda \alpha_i e_{N-1})_+ - (-x-\lambda \alpha_i e_{N-1})_+.
\end{equation}
In addition, 
$
\mathrm{prox}_{\lambda f_3^i}(x) = \mathrm{proj}_C(x),\ C=[0,1]^{N-1},
$
which is equivalent to
\begin{align*}
\left[\mathrm{prox}_{\lambda f_3^i}(x)\right]_j = \left[\mathrm{proj}_C(x)\right]_j=
\begin{cases}
1,& \mathrm{if}\ x_j>1,\\
x_j,& \mathrm{if}\ 0\leq x_j\leq1,\\
0,& \mathrm{if}\ x_j<0,
\end{cases}
\end{align*}
where ``proj" denotes the \textit{projection} operator.

The following lemma characterizes the aggregated proximal operator of functions $f_2^i$ and $f_3^i$ which is useful in designing the proximal algorithm.
\begin{lemma}\label{compositional_prox}
Functions $f_2^i$ and $f_3^i$ defined in \eqref{fs}, $\forall i\in\mathcal{N}$, satisfy the property:
$\mathrm{prox}_{\lambda (f_2^i + f_3^i)} = \mathrm{proj}_{C}\circ \mathrm{prox}_{\lambda f_2^i}$.
\end{lemma}
\begin{proof}
We proof for single dimensional case, i.e., $C=[0,1]$, and the analysis can be generalized for higher dimensional cases. By definition, we obtain
\begin{align*}
\mathrm{prox}_{\lambda (f_2^i + f_3^i)}(x) &= \arg\min_{l}\ f_2^i(l)+ f_3^i(l)+\frac{1}{2\lambda}\Vert l-x \Vert^2\\
& = \arg\min_{l\in C}\ f_2^i(l)+\frac{1}{2\lambda}\Vert l-x \Vert^2.
\end{align*}
Let $l^* = \arg_l\left(\frac{\partial \left( f_2^i(l)+\frac{1}{2\lambda}\Vert l-x \Vert^2\right)}{\partial x}=0\right)= \mathrm{prox}_{\lambda f_2^i }(x)$. In addition, function $f_2^i(l)+\frac{1}{2\lambda}\Vert l-x \Vert^2$ is decreasing in $l<l^*$ and increasing in $l\geq l^*$. Remind that $C=[0,1]$ is a closed set. Hence, when $0\leq l^* \leq 1$, $\mathrm{prox}_{\lambda (f_2^i + f_3^i)}(x)=l^*$; when $l<l^*$, $\mathrm{prox}_{\lambda (f_2^i + f_3^i)}(x)=0$; and when $l>l^*$, $\mathrm{prox}_{\lambda (f_2^i + f_3^i)}(x)=1$. In all three cases, we obtain 
$
\mathrm{prox}_{\lambda (f_2^i + f_3^i)}(x) = \mathrm{proj}_{C}(l^*) = \mathrm{proj}_{C}(\mathrm{prox}_{\lambda f_2^i }(x)).
$
\end{proof}

Lemma \ref{compositional_prox} indicates that we can deal with the convex terms of cognitive cost and feasible risk perception jointly. The security loss term $f_1^i$ is addressed in the ensuing section.

\subsection{Design of Proximal Algorithm}
Recall that $f^i_2$ and $f_3^i$, $\forall i\in\mathcal{N}$, are nonsmooth and not differentiable. To characterize the optimal cognition vector in $f^i_2$ and $f_3^i$, we first present the definition of subdifferential of a function which can be nonconvex and nonsmooth as follows.
\begin{definition}[Subdifferential \cite{rockafellar2009variational}]\label{subdifferential}
Let $f:\mathbb{R}^n\rightarrow \mathbb{R}$ be a proper and lower semicontinuous function.
\begin{enumerate}
\item The domain of $f$ is denoted by $\mathrm{dom}\ f:=\{x\in\mathbb{R}^n:f(x)<+\infty\}$.
\item For $x\in\mathrm{dom}\ f$, the Fr\'e{}chet subdifferential of $f$ at $x$ is the set of vectors $p\in\mathbb{R}^n$, denoted by $\hat{\partial}f(x)$, that satisfy
$$\underset{y\neq x, y\rightarrow x}{\lim\ \inf} \frac{1}{\Vert y-x \Vert}\left[ f(y)-f(x)- \langle p,y-x\rangle\right]\geq 0.$$
\item The limiting-subdifferential (or subdifferential) of $f$ at $x\in\mathrm{dom}\ f$, denoted by $\partial f(x)$, is defined by
\begin{align*}
\partial f(x):=\Big\{p\in\mathbb{R}^n:\exists x_n\rightarrow x, f(x_n)\rightarrow f(x),\\
p_k\in\hat{\partial}f(x_n)\rightarrow p\Big\}.
\end{align*}
\end{enumerate}
\end{definition}

\textit{Remark:} Based on the subifferential, a necessary condition for $x\in\mathbb{R}^n$ being a minimizer of $f$ is 
\begin{equation}\label{necessary_minimizer}
\partial f(x)\ni 0.
\end{equation}
Note that the points satisfying \eqref{necessary_minimizer} are called critical points of $f$. Our goal is to find a critical point $\bar{m}^i\in\mathrm{dom}\ Q_i$ that can be characterized by the necessary FOC: $0\in\partial Q_i(\bar{m}^i)$.

Note that $f_1^i$ is continuously differentialble with Lipschitz continuous gradient, i.e.,
$$\Vert \nabla f_1^i(x)- \nabla f_1^i(y)\Vert\leq L_i\Vert x-y \Vert, \ \forall x,y\in \mathbb{R}^{N-1},$$
where $L_i$ is the Lipschitz constant of $f_1^i$. Specifically, $\nabla f_1^i(m^i) = \Lambda^i m^i-\Lambda^i e_{N-1}$, which further yields
\begin{equation}
\Vert \nabla f_1^i(x)- \nabla f_1^i(y)\Vert = \Vert \Lambda^i(x-y)\Vert\leq L_i\Vert x-y \Vert, \ \forall x,y\in \mathbb{R}^{N-1}.
\end{equation}

The main steps in solving \eqref{P2} for a general $\Lambda^i$ of user $i$ are designed as follows:
\begin{align}
y_k^i &= x_k^i+\frac{t_{k-1}^i}{t_k^i}(z_k^i-x_k^i) + \frac{t_{k-1}^i-1}{t_k^i}(x_k^i-x_{k-1}^i),\label{y_k1}\\
z_{k+1}^i &= \mathrm{proj}_C\left(\mathrm{prox}_{\lambda_y^i f_2^i}(y_k^i- \lambda_y^i \nabla f_1^i(y_k^i))\right),\label{z_k1}\\
v_{k+1}^i &= \mathrm{proj}_C\left(\mathrm{prox}_{\lambda_x^i f_2^i}(x_k^i- \lambda_x^i \nabla f_1^i(x_k^i))\right),\label{v_k1}\\
t_{k+1}^i & = \left(1+\sqrt{4(t_k^i)^2+1}\right)/2,\label{t_k1}\\
x_{k+1}^i & = \begin{cases}
z_{k+1}^i, & \mathrm{if}\ Q_i(z_{k+1}^i)\leq Q_i(v_{k+1}^i),\\
v_{k+1}^i, & \mathrm{Otherwise},
\end{cases}\label{x_k1}
\end{align}
where the step constants $\lambda_x^i$ and $\lambda_y^i$ satisfy $0<\lambda_x^i<1/L_i$ and $0<\lambda_y^i<1/L_i$, respectively. If the algorithm converges, the values of $x_k^i,\ y_k^i,\ z_k^i$ and $v_k^i$ are the same which give the optimal cognition vector $m^i$.

\textit{Remark:} Note that \eqref{v_k1} serves as a monitor of the update in \eqref{z_k1}. Together with the condition in \eqref{x_k1}, each player updates their cognitive network when there is a sufficient decrease of the security management cost.

Before presenting the convergence results of the algorithm \eqref{y_k1}-\eqref{x_k1}, we first characterize a critical property of function $Q_i(m^i)$ defined in \eqref{P2}.

\begin{definition}[Kurdyka-\L{}ojasiewicz (KL) Property \cite{attouch2010proximal}]\label{KL_definition}
A function $f:\mathbb{R}^n\rightarrow (-\infty,+\infty]$ has the KL property at $x^*\in\mathrm{dom}\ \partial f:=\{x\in\mathbb{R}^n:\partial f(x)\neq\emptyset\}$ if there exists $\eta\in(0,+\infty]$, a neighborhood $U$ of $x^*$, and a desingularising function $\phi\in\Phi_{\eta}$, such that $\forall x\in U\cap\{x\in\mathbb{R}^n:f(x^*)<f(x)<f(x^*)+\eta\}$, the following KL inequality holds,
\begin{equation}\label{KL_inequality}
\phi'(f(x)-f(x^*))\mathrm{dist}(0,\partial f(x))\geq 1,
\end{equation}
where $\Phi_\eta$ includes a class of function $\phi:[0,\eta)\rightarrow \mathbb{R}^+$ satisfying: (1) $\phi$ is concave and $\phi\in C^1((0,\eta))$; (2) $\phi$ is continuous at $0$ with $\phi(0)=0$; and (3) $\phi'(x)>0,\ \forall x\in(0,\eta)$. In addition, $\mathrm{dist}(0,\partial f(x)):=\inf\left\{\Vert z\Vert:z\in \partial f(x)\right\}.$
\end{definition}

Note that a proper lower semicontinuous function $f$ having the KL property at each point of $\mathrm{dom}\ \partial f$ is called a \textit{KL function}. KL inequality \eqref{KL_inequality} ensures that, by choosing a proper desingularising function $\phi$, we can reparameterize the values of function $f$ near its critical points to avoid flatness. Thus, $\phi$ has an impact on the convergence rate of the designed algorithm which will be presented in Theorem \ref{convergence}.   KL property is general in functions. Notably, the semi-algebraic functions satisfy the KL property \cite{attouch2010proximal}. Some examples include real polynomial functions, indicator functions of semi-algebraic sets and $\Vert\cdot\Vert_{p}$ with $p\geq 0$. Furthermore, the semi-algebraic property preserves under composition, finite sums and products of semi-algebraic functions  \cite{bolte2014proximal}.

\begin{lemma}\label{KL_f}
Functions $f_1^i,\ f_2^i$ and $f_3^i$ in \eqref{fs} satisfy the KL property, and thus $Q_i$ in \eqref{P2} is a KL function. In addition, the desingularising function $\phi(u)$ can be chosen as $\phi(u)=\frac{\kappa}{\theta}u^{\theta}$ for some $\theta\in(0,\frac{1}{2}]$  and $\kappa>0$.
\end{lemma}
\begin{proof}
We know that $f_1^i,\ f_2^i$ and $f_3^i$ are semi-algebraic functions, and thus $Q_i$ satisfies the KL property \cite{attouch2010proximal}. Remind that when $m^i\notin C:=\{m^i|0\leq m^i_j\leq 1,j\neq i,j\in\mathcal{N}\}$, $Q_i(m^i)\rightarrow +\infty$. Based on Definition \ref{subdifferential}, we obtain  $\mathrm{dom}\ \partial Q_i = C$. Therefore, $Q_i(m^i)$ is analytic over $\mathrm{dom}\ \partial Q_i$. In addition, the desingularising function of real-analytical functions satisfying inequality \eqref{KL_inequality} can be chosen as $\phi(u)=u^{1-\delta}$, where $\delta\in[\frac{1}{2},1)$ \cite{bolte2014proximal}.
\end{proof}

Based on Lemma \ref{KL_f}, we present the convergence result of the designed algorithm \eqref{y_k1}-\eqref{x_k1} in Theorem \ref{convergence}.

\begin{theorem}\label{convergence}
The algorithm given by \eqref{y_k1}-\eqref{x_k1} converges to a critical point with rates related to the parameters $\kappa$ and $\theta$, where $\kappa$ and $\theta$ are defined in Lemma \ref{KL_f}. Specifically, there exists a $k_0$ such that $\forall k>k_0$,
$$
Q_i(x_k)-Q_i^*\leq \left(\frac{\kappa}{(k-k_0)(1-2\theta)d_2}\right)^{\frac{1}{1-2\theta}},
$$ 
where $Q_i^*$ is the function value achieved at critical points of $\{x_k\}$, $d_2=\min\left\{  \frac{1}{2d_1\kappa},\sigma (Q_i(v_0)-Q_i^*)^{2\theta-1}\right\}$,  $d_1=2\alpha(\frac{1}{\lambda_x}+L)^2/(1-2\alpha)$, and $\sigma=\frac{\kappa}{1-2\theta}\left(2^{\frac{2\theta-1}{2\theta-2}}-1\right)$.
\end{theorem}
\begin{proof}
See Appendix \ref{thm1}.
\end{proof}

For a special case where $f_1^i$ is convex, the following simplified steps \eqref{y_k2}-\eqref{x_k2} can be adopted to accelerate the computation. The monitoring update step $v_{k+1}$ is omited due to the convexity of $f_1^i$. This algorithm is slightly different with the one in \cite{beck2009fast} in terms of the projection step. Since $Q_i$ is convex, then algorithm \eqref{y_k2}-\eqref{x_k2} converges to a unique optimal solution.
\begin{align}
y_k^i &= x_k^i+\frac{t_{k-1}^i}{t_k^i}(z_k^i-x_k^i) + \frac{t_{k-1}^i-1}{t_k^i}(x_k^i-x_{k-1}^i),\label{y_k2}\\
z_{k+1}^i &= \mathrm{proj}_C\left(\mathrm{prox}_{\lambda_y^i f_2^i}(y_k^i- \lambda_y^i \nabla f_1^i(y_k^i))\right),\\
t_{k+1}^i & = \left(1+\sqrt{4(t_k^i)^2+1}\right)/2,\\
x_{k+1}^i & = \begin{cases}
z_{k+1}^i, & \mathrm{if}\ Q_i(z_{k+1}^i)\leq Q_i(x_{k}^i),\\
x_{k}^i, & \mathrm{Otherwise}.
\end{cases}\label{x_k2}
\end{align}
Similar to \eqref{y_k1}-\eqref{x_k1}, when the algorithm \eqref{y_k2}-\eqref{x_k2} converges, the values of $x_k^i,\ y_k^i$ and $z_k^i$ are the same which give the optimal cognition vector $m^i$.

\begin{algorithm}[!t]
\caption{Cognitive Network Formation for Player $i$}
\label{algorithm1} 
\begin{enumerate}
\item Input $f_1^i$, $f_2^i$ and $C=[0,1]^{N-1}$
\item Initialize parameters $z_0^i$, $x_0^i$, $x_1^i$, $t_0^i$, $t_1^i$, $\lambda_x^i$ and $\lambda_y^i$ 
\item \textbf{for} $k=1,2,...$ \textbf{do}
\item \quad \textbf{if} $f_1^i$ is convex
\item \qquad Update $y_k^i$, $z_{k+1}^i$, $v_{k+1}^i$, $t_{k+1}^i$ and $x_{k+1}^i$ through \eqref{y_k2}-\eqref{x_k2}
\item \quad \textbf{else}
\item \qquad Update $y_k^i$, $z_{k+1}^i$, $v_{k+1}^i$, $t_{k+1}^i$ and $x_{k+1}^i$ through \eqref{y_k1}-\eqref{x_k1}
\item \quad \textbf{end}
\item \textbf{end for}
\item \textbf{Return} $m^i = x_k^i$
\end{enumerate}
\end{algorithm}

\textit{Homogeneous Users Case:} When the agents in the IoT network are homogeneous, i.e., $R_{ii}^i=R_{jj}^j$, $R_{ij}^i = R_{ji}^j$, $r_i=r_j=r$, $\beta_i=\beta_j=\beta\leq N-1$, $\forall i,j\in\mathcal{N}$, we can characterize the closed form solutions of decisions $u_i$ and  $m^i$, $\forall i\in\mathcal{N}$. Specifically, we obtain, $\forall i\in\mathcal{N}$,
\begin{equation}\label{homo_solution}
\begin{split}
m^{i*}_{j}&= \frac{\beta}{N-1},\ \forall j\neq i, j\in\mathcal{N},\\
 u_i^* &= \frac{r}{R_1-\beta R_2},
 \end{split}
\end{equation}
where $R_1 = R_{ii}^i$ and $R_2 = R_{jk}^i$ for $j\neq i$ and $k\neq i$. The results indicate that, at GNE, the cognitive network that each user $i$ forms, $i\in\mathcal{N}$, is symmetric, i.e., the allocated attention to other users $j\neq i$ by user $i$ is the same. In addition, with a larger $\beta$, the users spend more effort on the security management at GNE. This can be interpreted as follows: with a better perception of cyber risks in the IoT, the users becomes better informed of the risks and make best effort to reduce the security loss.

\begin{algorithm}[!t]
\caption{Strategic Risk Management with Bounded Rationality}
\label{algorithm2} 
\begin{enumerate}
\item Initialize parameters in the game $\mathcal{G}$, cognition cost $\alpha_i$, cognitive networks $m^i,\ \forall i\in\mathcal{N}$
\item \textbf{Do}
\\ \textbf{Best response dynamics:} 
\item Based on $m^i$, $i\in\mathcal{N}$, player $i$ determines their best-response strategy through \eqref{sparse_BR_dynamic} iteratively until reaching a BRNE
\\ \textbf{Cognitive network formation:} 
\item Each player $i$, $i\in\mathcal{N}$, forms their cognitive network $m^i$ through Algorithm \ref{algorithm1}
\item \textbf{Until} $[m^i]_{i\in\mathcal{N}}$ and $[u_i]_{i\in\mathcal{N}}$ converge
\item \textbf{Return} $m^i$ and $u_i$, $\forall i\in\mathcal{N}$, which form a GNE
\end{enumerate}
\end{algorithm}

\subsection{Integrated Algorithm and Discussions}\label{integrated_algorithm}
For clarity, we summarize the combined algorithm including the strategic security decision-makings of players in the IoT networks and their corresponding cognitive network formations together in Algorithm \ref{algorithm2}. The integrated algorithm exhibits an alternating pattern between the best-response of security management and the strategic cognitive network formation of IoT users.

We next discuss some observations obtained from the algorithm. The steps $z_{k+1}^i$ and $v_{k+1}^i$ in \eqref{z_k1} and \eqref{v_k1} of the algorithm can be simplified further. Here, we only analyze $z_{k+1}^i$, and the procedure follows for $v_{k+1}^i$. First, we have $\nabla f_1^i(y_k^i) = \Lambda^i(y_k^i - e_{N-1})$. Then, $[y_k^i- \lambda_y^i \nabla f_1^i(y_k^i)]_j = [y_k^i- \lambda_y^i\Lambda^i(y_k^i - e_{N-1})]_j\geq 0$, $\forall j\neq i, j\in\mathcal{N}$. Thus, based on \eqref{prox_f2}, we obtain
\begin{align*}
z_{k+1}^i
&= \mathrm{proj}_C\left(y_k^i- \lambda_y^i\Lambda^i(y_k^i - e_{N-1})- \lambda_y^i\alpha_i e_{N-1}\right)\\
&= \mathrm{proj}_C\left(y_k^i+ \lambda_y^i(\Lambda^i(e_{N-1}- y_k^i)- \alpha_i e_{N-1})\right).
\end{align*}
The update of player $i$'s attention on player $j$ at step $k+1$, $j\neq i$, can be expressed as
\begin{align*}
&[z_{k+1}^i]_j=\\
& \mathrm{proj}_{[0,1]}\left([y_k^i]_j+ \lambda_y^i\left(\frac{R_{ij}^i}{R_{ii}^i} u_j \sum_{p\neq i,p\in\mathcal{N}}R_{ip}^i u_p \left(1-[y^i_k]_p\right) - \alpha_i \right)\right).
\end{align*}

When $\frac{R_{ij}^i}{R_{ii}^i} u_j \sum_{p\neq i,p\in\mathcal{N}}R_{ip}^i u_p \left(1-[y^i_k]_p\right) \geq \alpha_i$ which is equivalent to $\sum_{p\neq i,p\in\mathcal{N}}R_{ip}^i u_p [y^i_k]_p \leq \sum_{p\neq i,p\in\mathcal{N}}R_{ip}^i u_p - \frac{R_{ii}^i}{R_{ij}^i u_j} \alpha_i$, we know that $[z_{k+1}^i]_j\geq [z_{k}^i]_j$. The player $i$'s attention on player $j$ increases at step $k+1$, since there remains extra cognition resources to be allocated which corresponds to a phenomenon called \textit{filling the inattention}. In addition, a smaller cognition cost $\alpha_i$ yields a larger upper bound for $\sum_{p\neq i,p\in\mathcal{N}}R_{ip}^i u_p [y^i_k]_p$, and hence player $i$ can pay more attention to other players which again leads to the observation of filling the inattention. 

In the IoT network, user $j$'s decision has an impact on the strategy of user $i$. To illustrate the discovery, we consider two groups of IoT users, and one group of users have more incentive to secure the devices, i.e., their security investment is larger. Then, from user $i$'s perspective, his attention on user $j$ is influenced by the term ${R_{ii}^i}/({R_{ij}^i u_j})$. When user $j$ lies in the group of a higher investment $u_j$, then the upper bound $\sum_{p\neq i,p\in\mathcal{N}}R_{ip}^i u_p - \frac{R_{ii}^i}{R_{ij}^i u_j} \alpha_i$ is larger. Therefore, each IoT user will allocate more cognition resources to the users in the group with a higher security standard which exposes the phenomenon of \textit{emergence of partisanship}. 

In a heterogeneous IoT network, the system parameters $R_{ij}^i$, $R_{ii}^i$, and decisions $u_i$ are generally different. Then, for player $i\in\mathcal{N}$, the term $\frac{R_{ij}^i}{R_{ii}^i} u_j \sum_{p\neq i,p\in\mathcal{N}}R_{ip}^i u_p$, $j\neq i,\ j\in\mathcal{N}$, identifies the most influential agents in the network. Moreover, the critical agents to pay attention to for each user almost overlap, resulting the phenomenon of \textit{attraction of the mighty} during the cognitive network formation.

We illustrate the discovered phenomena in Section \ref{simulation}.



%

\section{Case Studies}\label{simulation}
We use case studies of IoT-enabled smart communities shown in Fig. \ref{SH} to corroborate the designed algorithms and illustrate the security management of bounded rational agents in this section.

\subsection{Effectiveness of Algorithm \ref{algorithm1}}
First, we verify the effectiveness of Algorithm \ref{algorithm1}. Specifically, we choose $N=10$, $\alpha = 100$ and generate a $9\times 9$ random matrix which is not definite for $\Lambda^i$. Thus, $f_2^i$ in \eqref{fs} is not convex. The iterative updates through the designed proximal algorithm are presented in Fig. \ref{algorithm1_converge} which reveal fast convergence to the steady state. In addition, the algorithm yields a sparse cognition vector $m = [1,0,0,0,0.41,1,0,0.30,0.26]$. To investigate the robustness of the algorithm, we study the same network as in Fig. \ref{10_agents_1} with different initial conditions. The results are shown in Figs. \ref{10_agents_2} and \ref{10_agents_3}. We can verify that the steady states in Figs. \ref{10_agents_2} and \ref{10_agents_3} are the same as the ones in Fig. \ref{10_agents_1} which corroborate the robustness of the algorithm to initial conditions. To further verify the algorithm, we also investigate the network containing different numbers of agents. The results with 7 and 15 agents are presented in Figs. \ref{7_agents} and \ref{15_agents}. Both results indicate that the designed algorithm is reliable in computing the sparse steady strategy.  After conducting sufficient number of case studies, we conclude that the algorithm is effective with probability 1 under arbitrary number of agents.


\begin{figure}[t]
  \centering
    \subfigure[$N=10$, case 1]{
    \includegraphics[width=0.8\columnwidth]{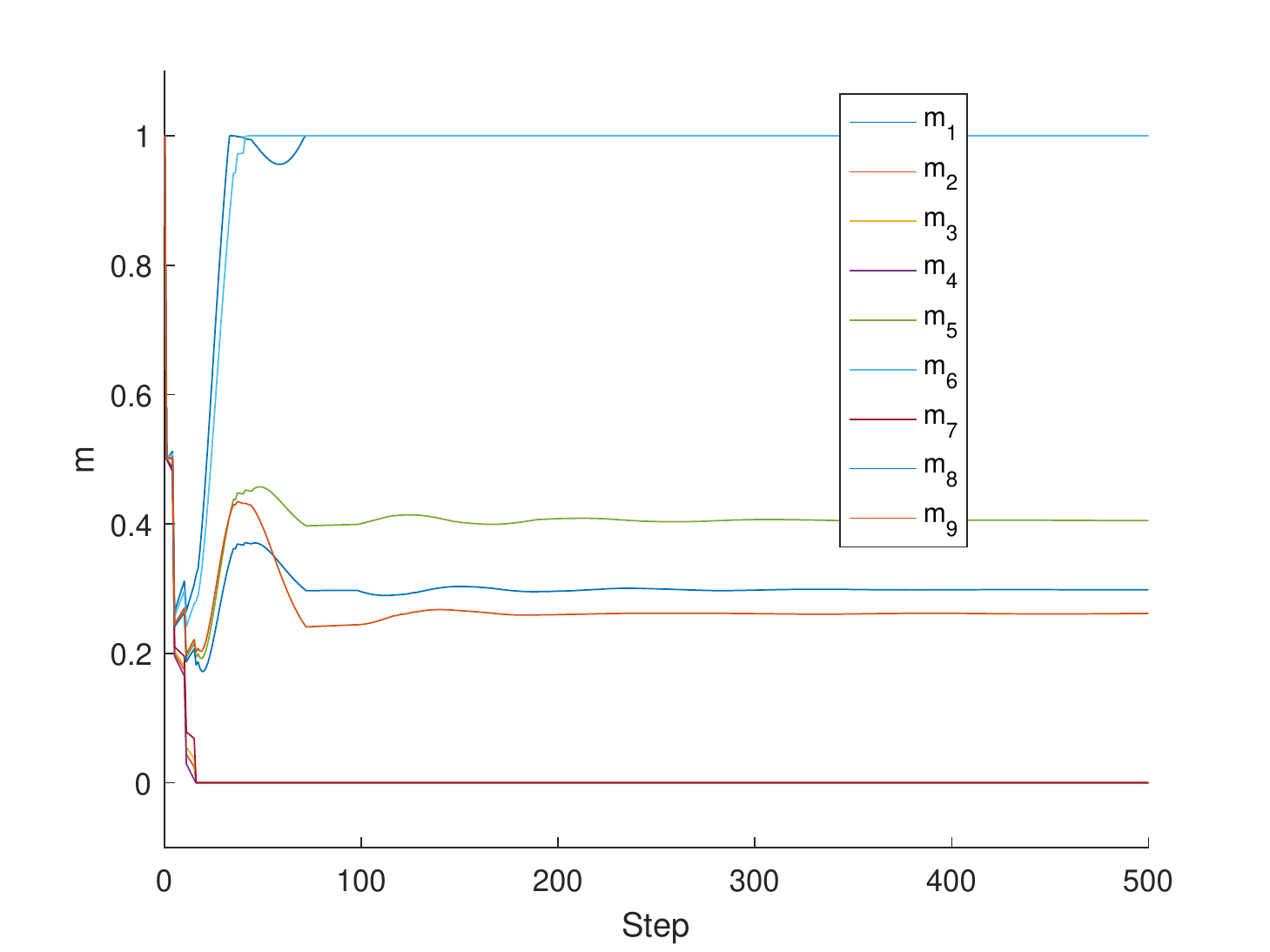}\label{10_agents_1}}
        \subfigure[$N=10$, case 2]{
    \includegraphics[width=0.48\columnwidth]{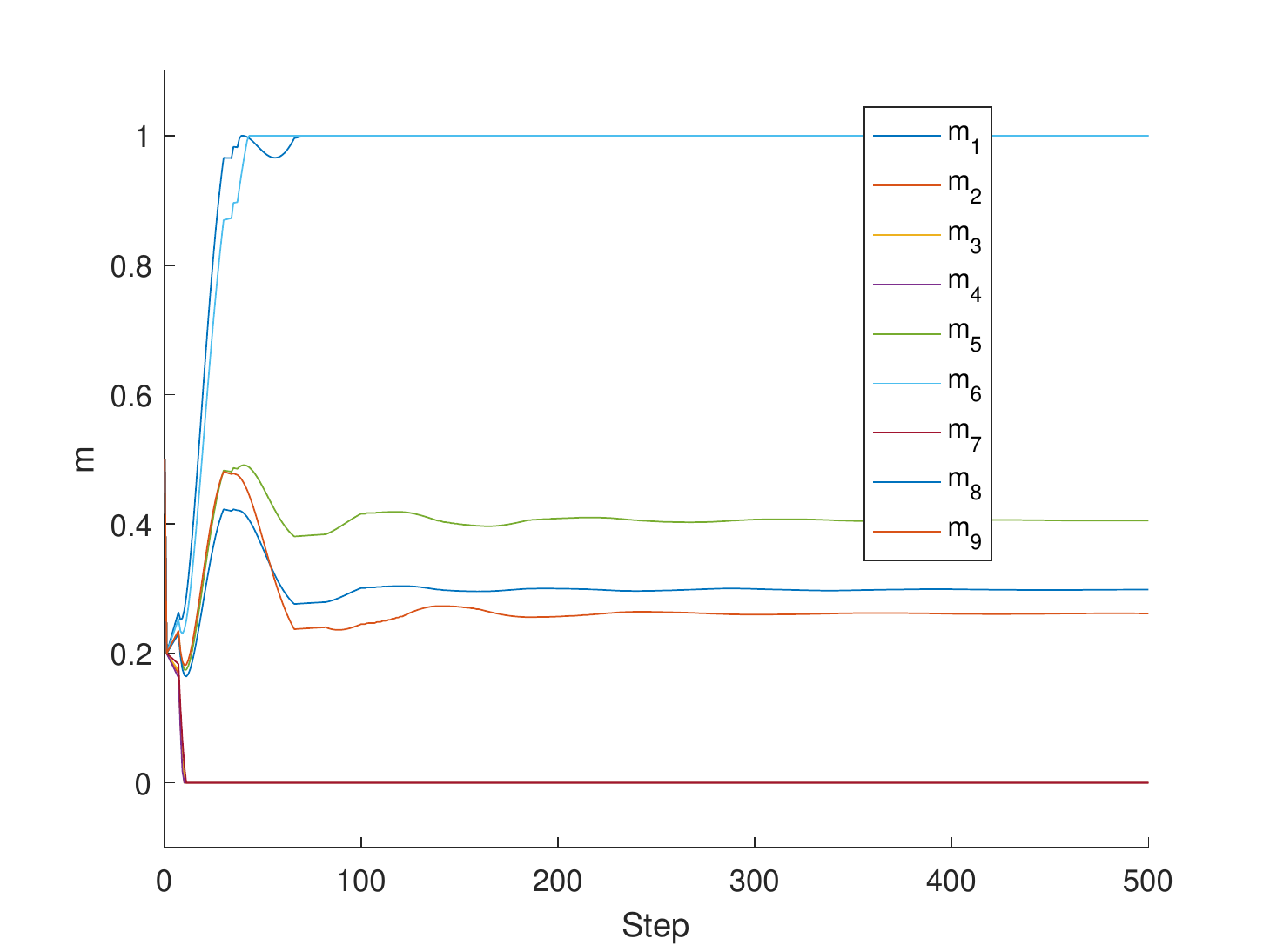}\label{10_agents_2}}
        \subfigure[$N=10$, case 3]{
    \includegraphics[width=0.48\columnwidth]{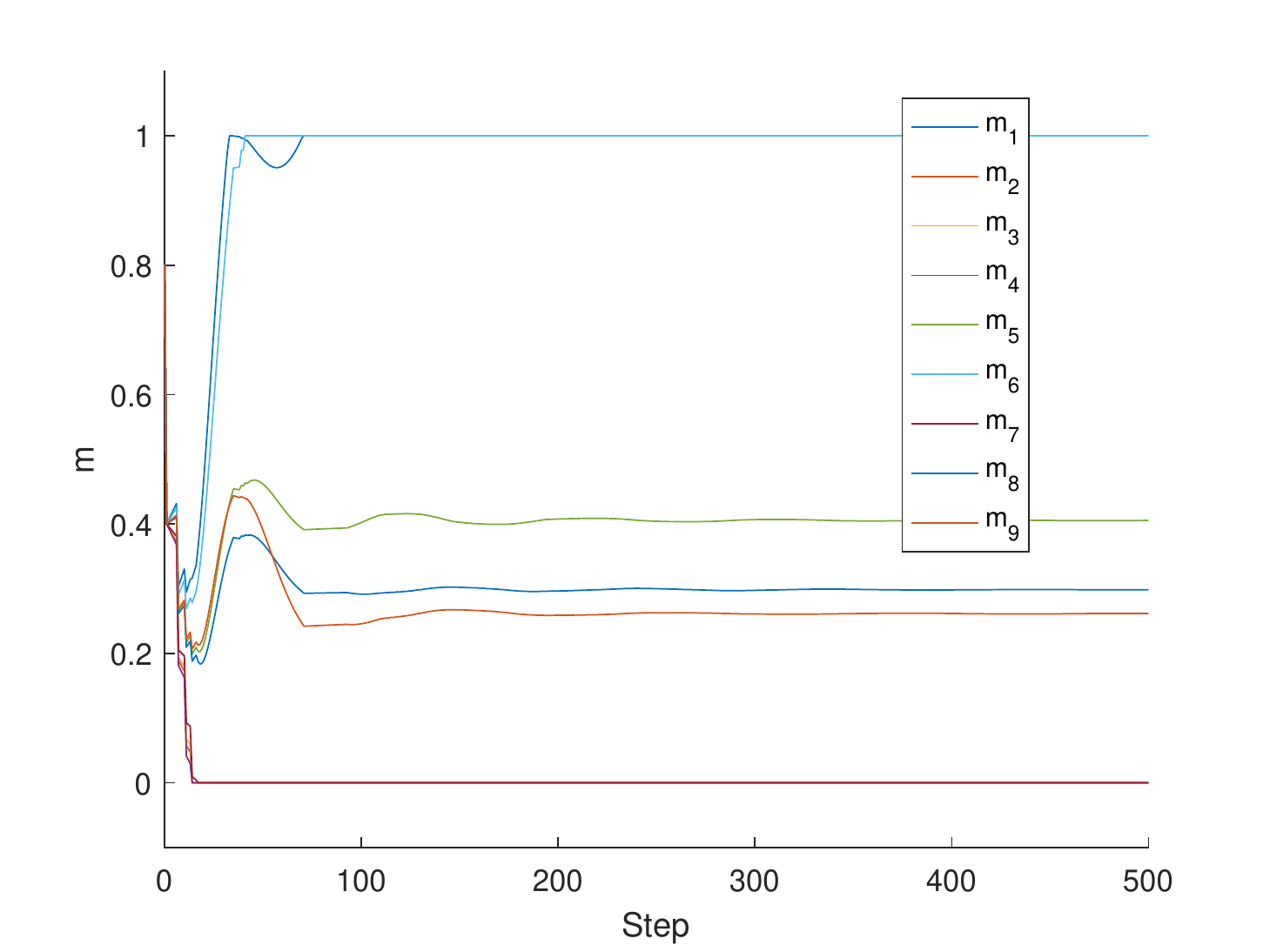}\label{10_agents_3}}
  \subfigure[$N=7$]{
    \includegraphics[width=0.48\columnwidth]{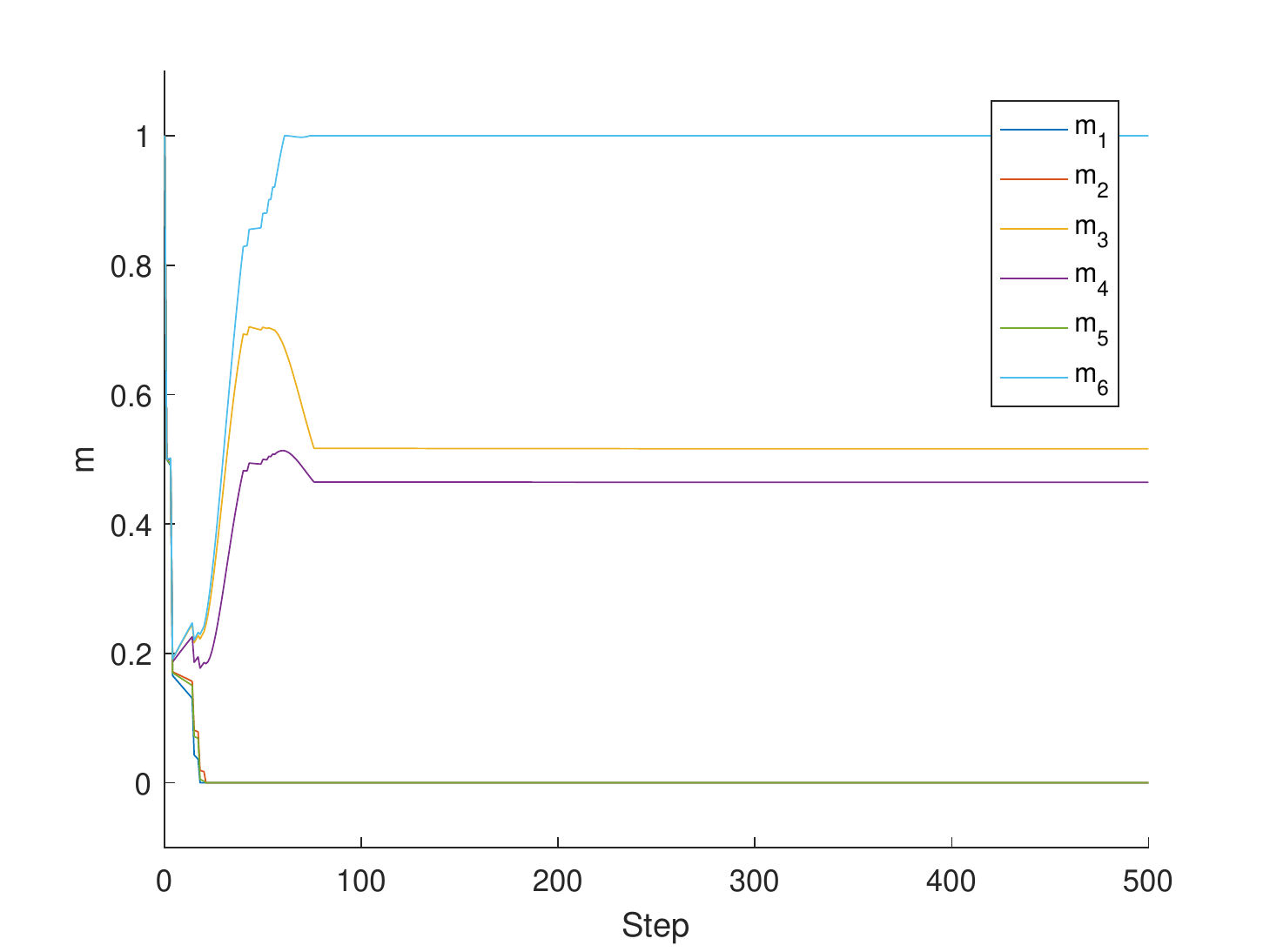}\label{7_agents}}
	 \subfigure[$N=15$]{
    \includegraphics[width=0.48\columnwidth]{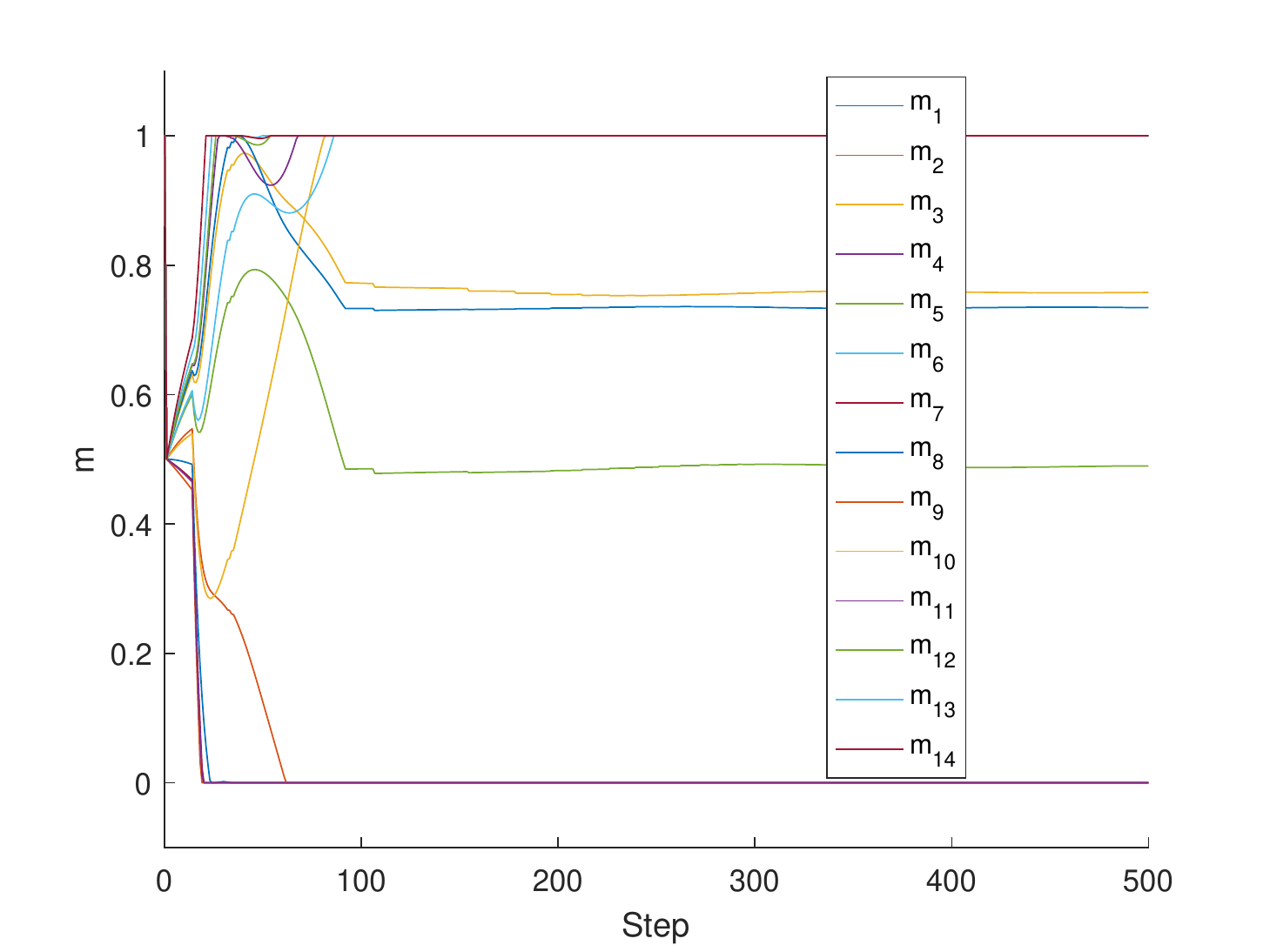}\label{15_agents}}
  \caption[]{Performance of Algorithm \ref{algorithm1} on a nonconvex $f_2^i$ in \eqref{fs}. (a), (d), and (e) show the results with 10, 7 and 15 agents in the network, respectively. The network configurations in (a), (b), and (c) are the same, but their initial conditions are different. The algorithm yields the same result for cases in (a), (b), and (c) which shows the robustness of the algorithm. }
  \label{algorithm1_converge}
\end{figure}

\subsection{Homogeneous Smart Homes}\label{homo_case}
In this case study, we consider $N=10$ homogeneous households in the smart community, i.e., all the parameters are the same for each agent. Specifically, the parameters are chosen as follows: $R_{jk}^i = 20 \ \mathrm{unit}/\mathrm{k}\$^2$ if $j=k=i$, otherwise $R_{jk}^i = 1\ \mathrm{unit}/\mathrm{k}\$^2$, $\forall i$; $r_i=25\ \mathrm{unit}/\mathrm{k}\$$, $\forall i\in\mathcal{N}$ and $\alpha^i=\alpha$, $\forall i$, is chosen to satisfy $\beta = \Vert m^i\Vert_1=3$. The selected parameters indicate that the security level of a household is mainly determined by its own security management policy rather than the ones of connected households. Recall that we have obtained the analytical solutions for homogeneous case in \eqref{homo_solution} which yield $m^i_j = \frac{1}{3},\ \forall j\neq i,j\in\mathcal{N}$ and $u_i = \frac{25}{17}=1.47\mathrm{k}\$$. Thus, each agent allocates attention resource equally to their connected neighbors.  Fig. \ref{case1} presents the results by using Algorithm \ref{algorithm2}, where the step index represents an iteration between two components of cognitive network formation and security investment. We can conclude that the rational decision yields less cost for players compared with their irrational decision counterparts due to the bounded rationality. Furthermore, the algorithm gives the same cognitive network structure as the one obtained from the analytical results which corroborates the proposed integrated algorithm.

\begin{figure}[t]
  \centering
  \subfigure[rational strategy]{
    \includegraphics[width=0.48\columnwidth]{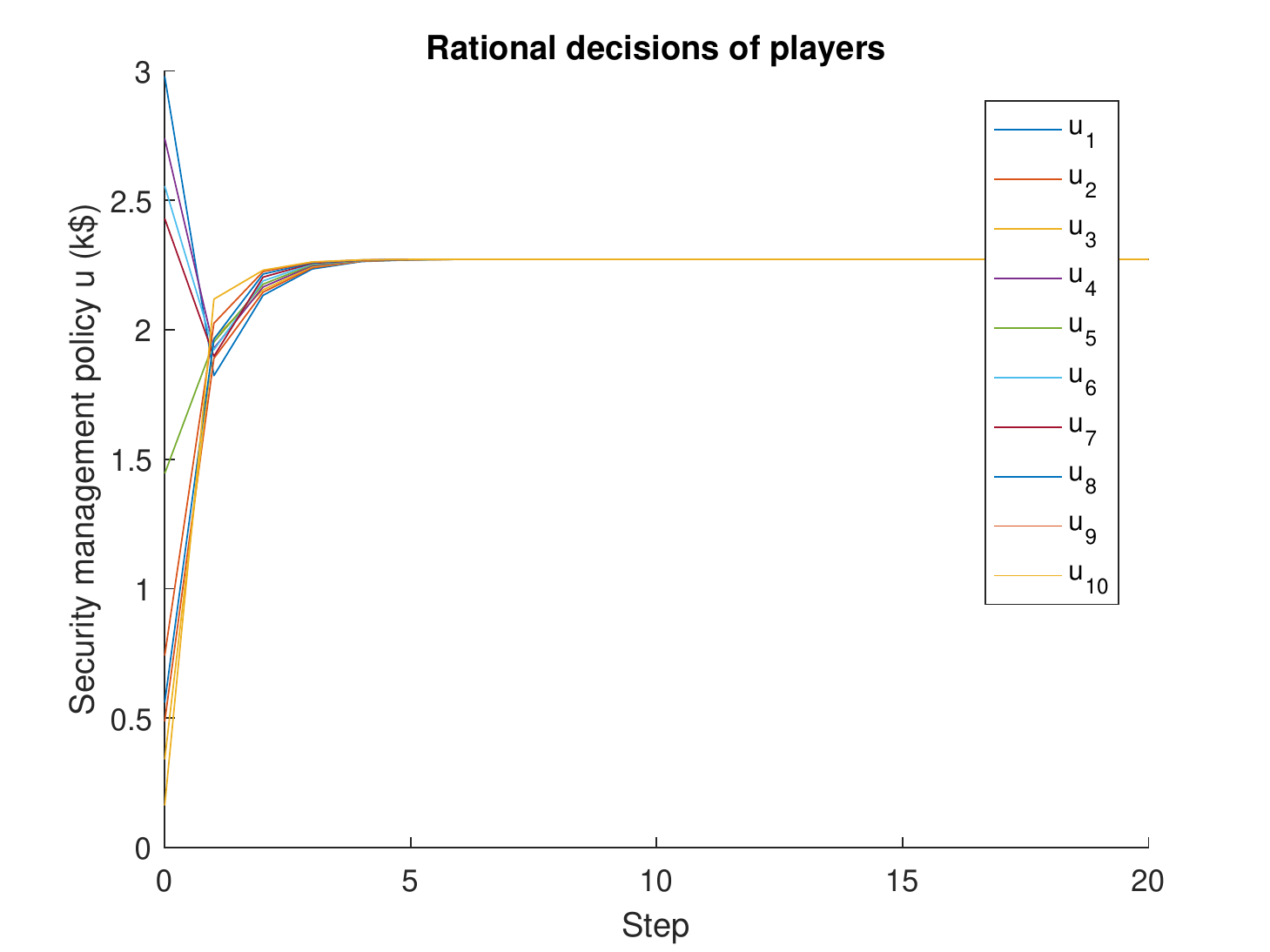}}
	 \subfigure[cost under rational strategy]{
    \includegraphics[width=0.48\columnwidth]{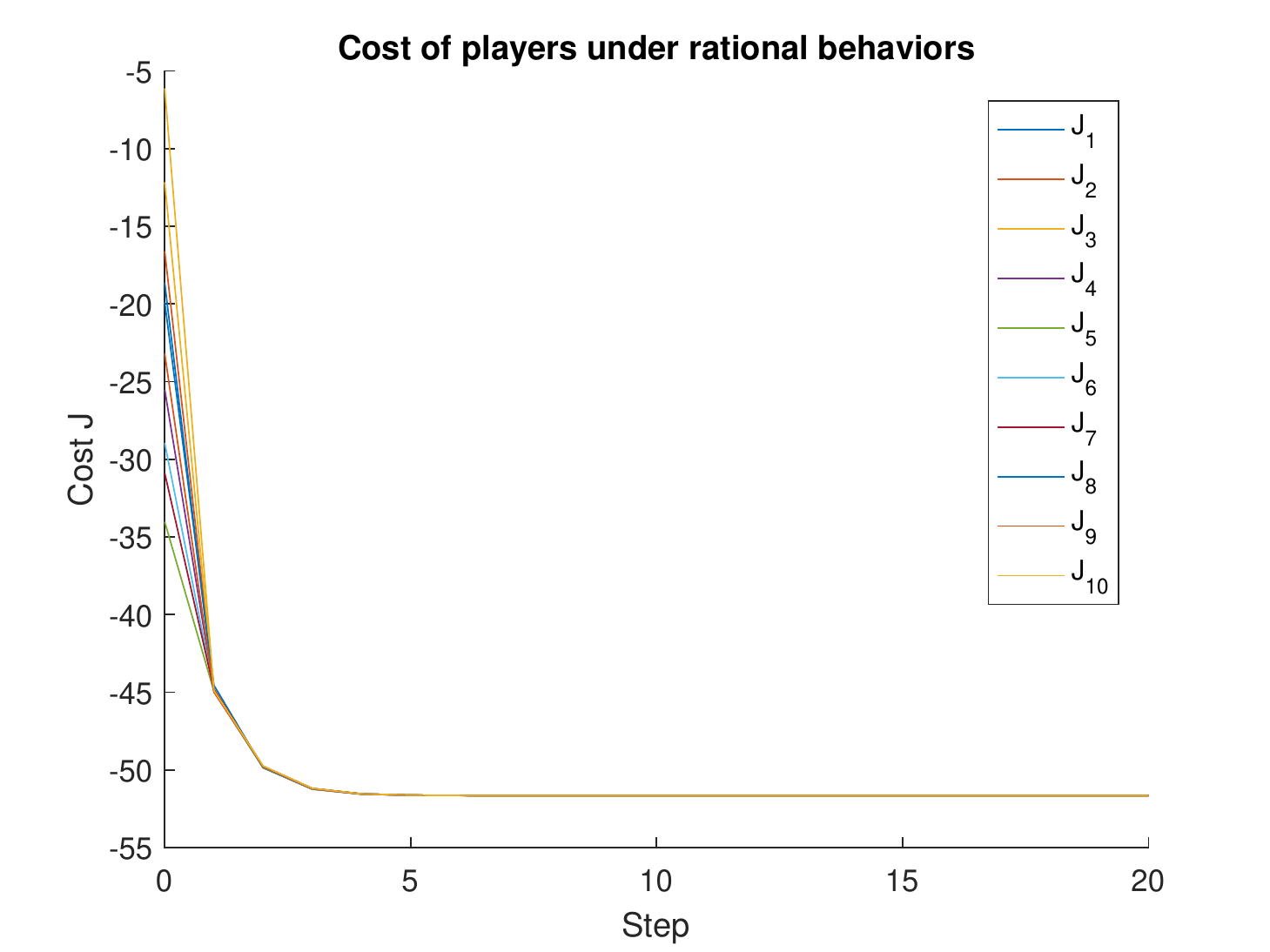}}
    	 \subfigure[bounded rational strategy]{
    \includegraphics[width=0.48\columnwidth]{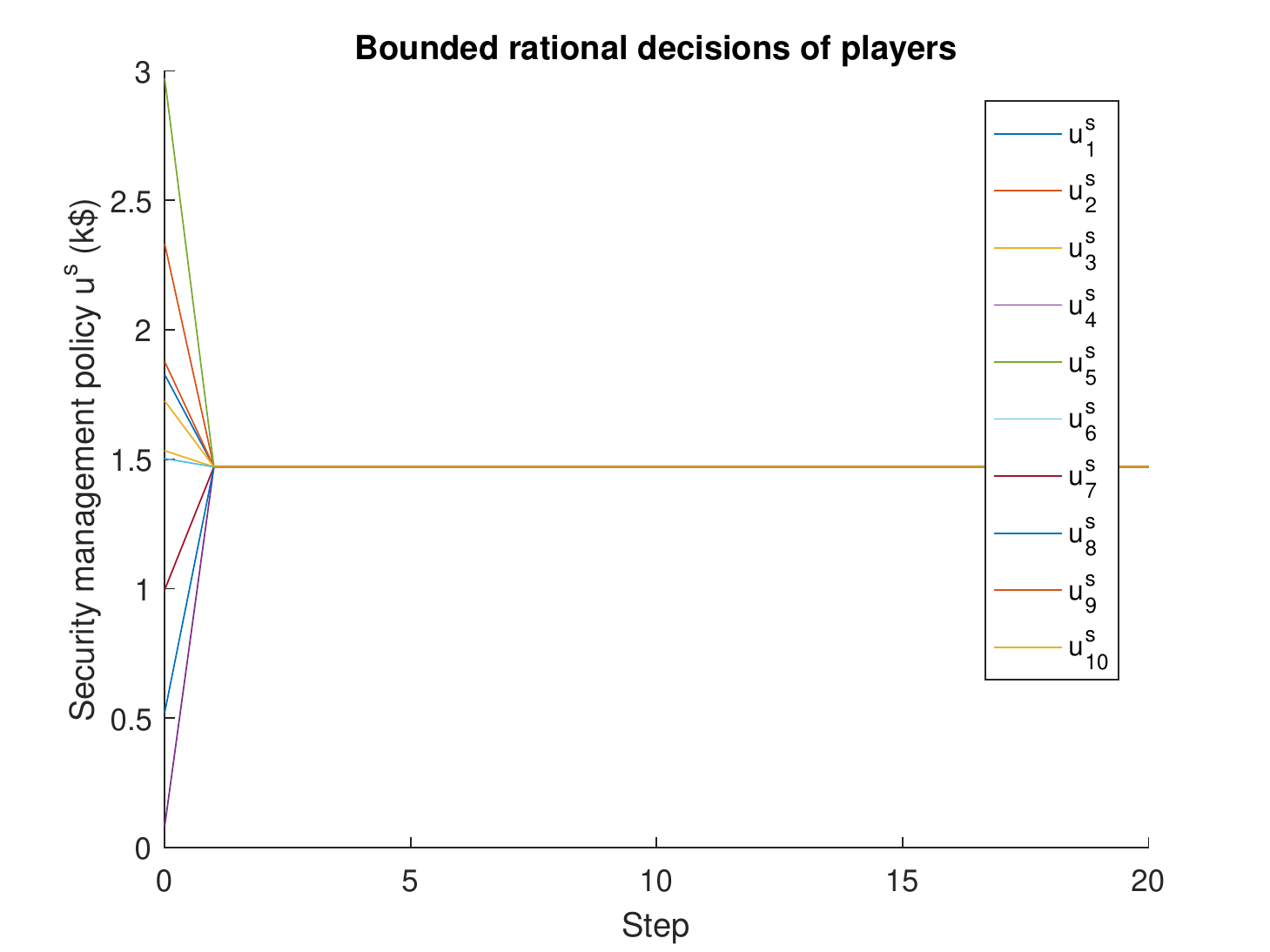}}
    	 \subfigure[cost under bounded rational strategy]{
    \includegraphics[width=0.48\columnwidth]{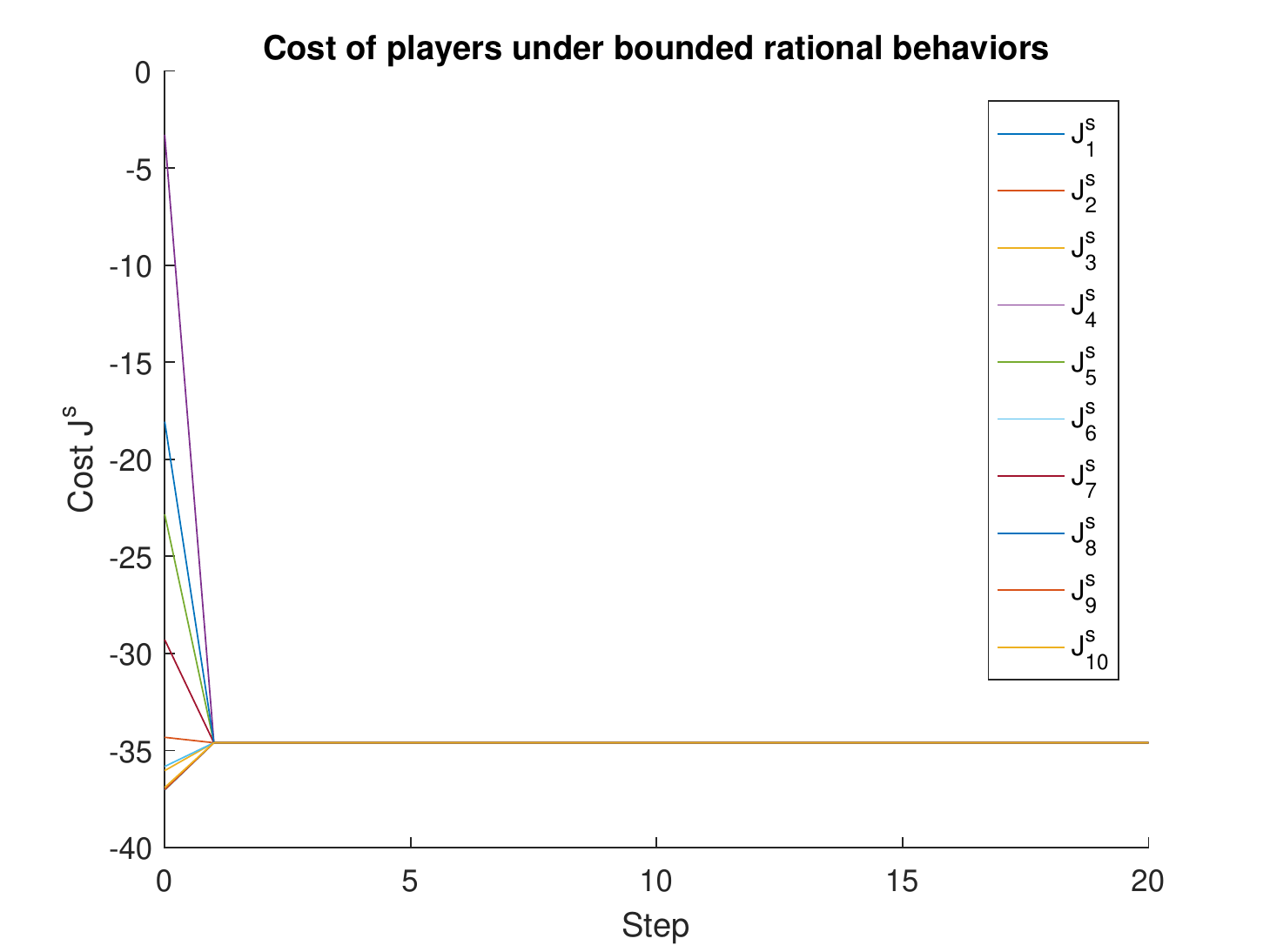}}
        	 \subfigure[cognitive network]{
    \includegraphics[width=0.7\columnwidth]{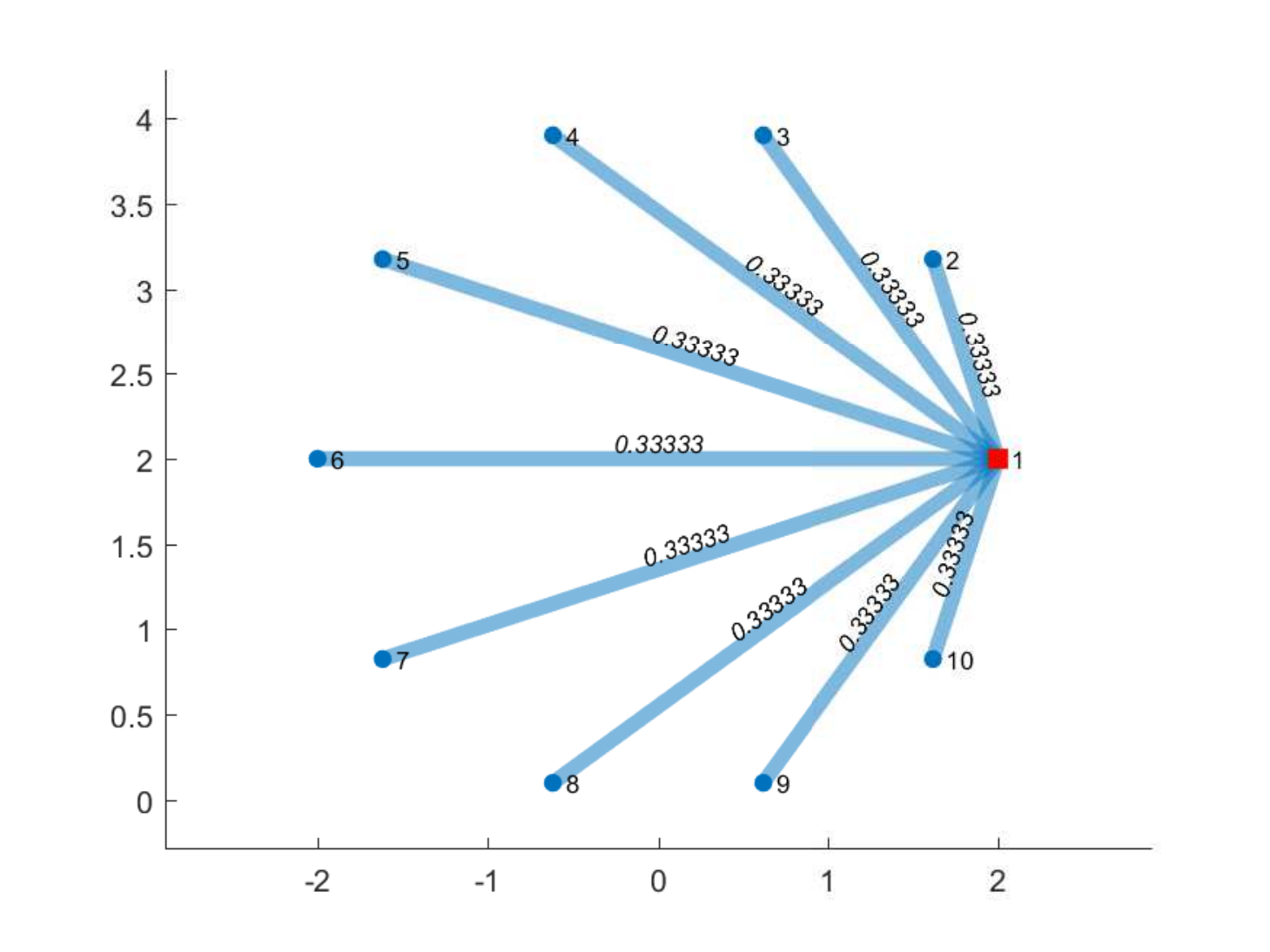}}
  \caption[]{(a) and (b) are the rational decision of players and the corresponding cost, respectively. (c) and (d) are the counterparts of (a) and (b) with bounded rationality. (e) illustrates the formed cognitive networks which is symmetric in this homogeneous case.}
  \label{case1}
\end{figure}

\subsection{Emergence of Partisanship}\label{two_group}
We next investigate a smart community including two groups of agents denoted by $G1$ and $G2$, respectively. Specifically, $G1$ includes 5 agents, $G1=\{1,...,5\}$, and $G2$ contains 10 agents, $G2=\{6,...,15\}$.  The parameters are the same as those in Section \ref{homo_case} except that for agents in $G1$, $r_i=40\ \mathrm{unit}/\mathrm{k}\$$, $\forall i\in G1$, to distinguish two groups of users. Thus, the agents in $G1$ have more incentives to secure their IoT products than those in $G2$. 
Fig. \ref{case2} shows the results.
For agents in $G1$, the cognitive network is characterized by $m^i = [0.75,...,0.75,0,...,0]$, $i\in G1$, and for agents in $G2$, $m^j = [0.6,...,0.6,0,...,0]$, $j\in G2$. Therefore, with limited cognition, all agents only pay attention to the security decisions made by smart homes in $G1$ which yields the phenomenon of partisanship. We also verify that the RBP increases due to the bounded rationality.

\subsection{Filling the Inattention}
Under the setting of Section \ref{two_group}, we further assume that the agents have better cognitive ability and can perceive more cyber risks in the smart community in a way that $\beta = \Vert m^i\Vert_1=8$.  Other parameters are the same as those in Section \ref{two_group}. Fig. \ref{case3} presents the results. Specifically, we obtain $m^i = [1,...,1,0.4,...,0.4]$ for $i\in G1$ and  $m^j = [1,...,1,0.33,...,0.33]$ for $j\in G2$, which show that the agents in $G1$ play a critical role in the security risk management of smart community. Furthermore, with more cognition resource, the agents in $G2$ that are not paid attention to previously in Section \ref{two_group} are considered by other households. This phenomenon is termed as filling the inattention.

\begin{figure}[t]
  \centering
    	 \subfigure[bounded rational strategy]{
    \includegraphics[width=0.48\columnwidth]{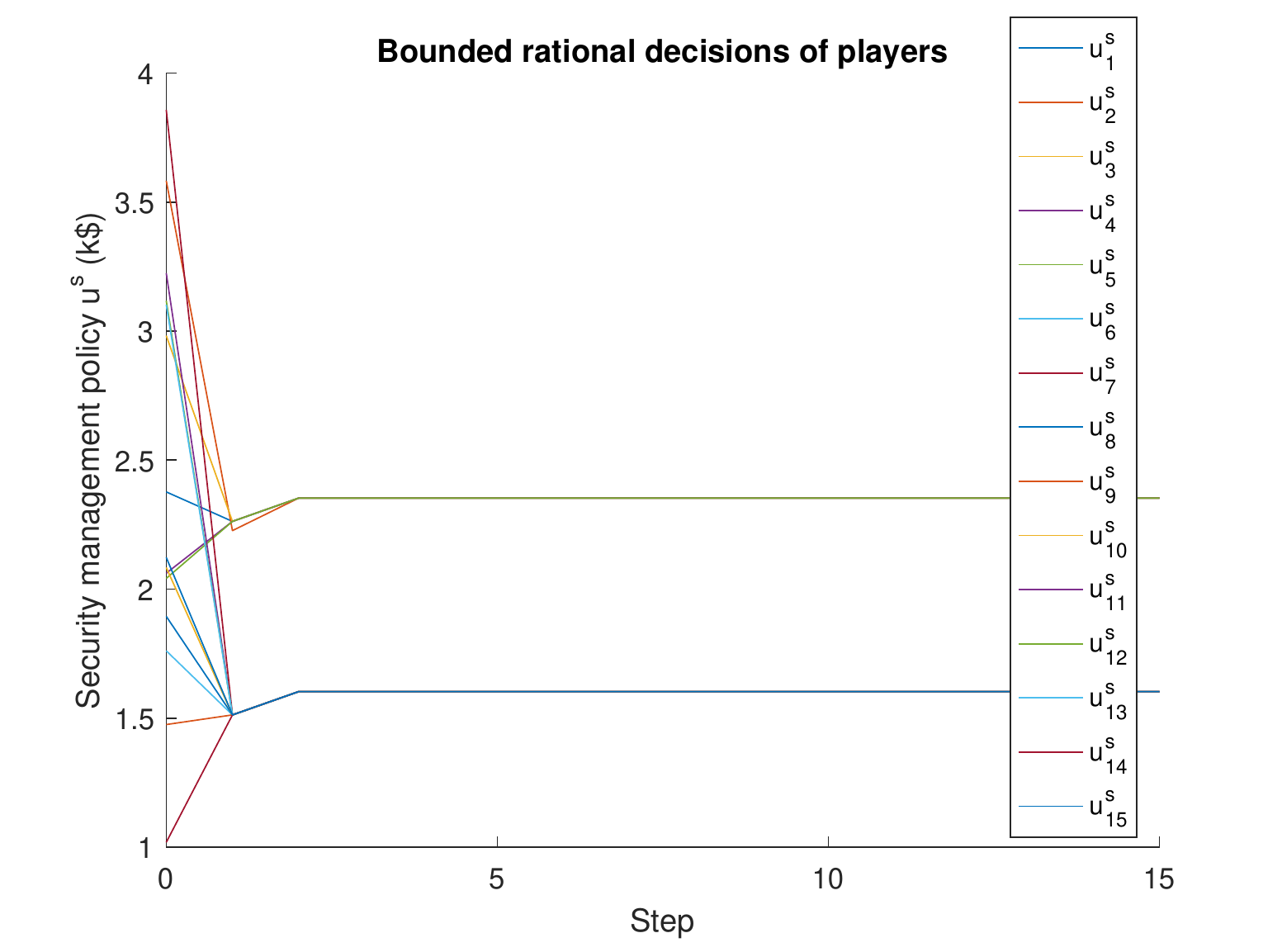}}
    	 \subfigure[RBP]{
    \includegraphics[width=0.48\columnwidth]{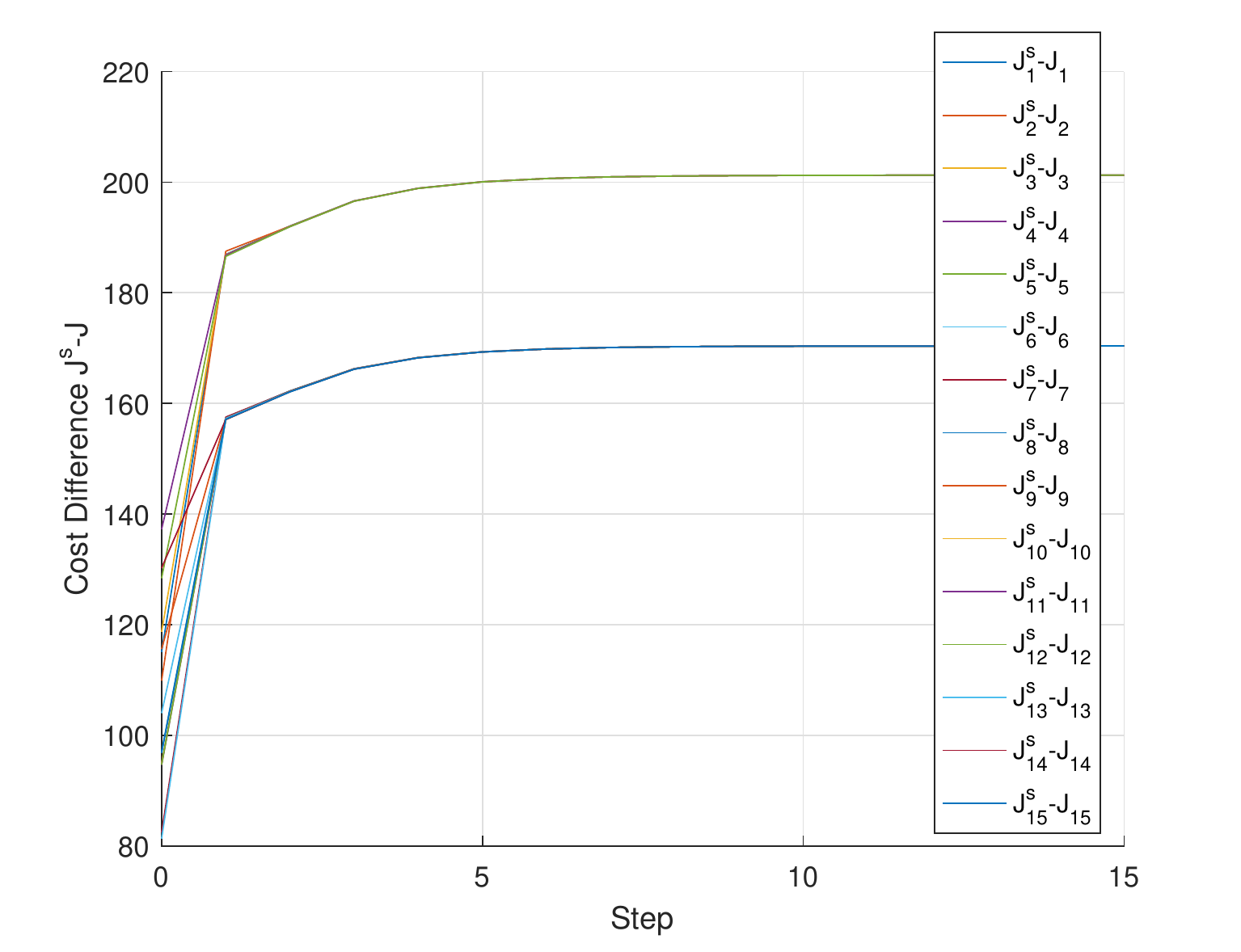}}
        	 \subfigure[cognition vector $m$]{
    \includegraphics[width=0.48\columnwidth]{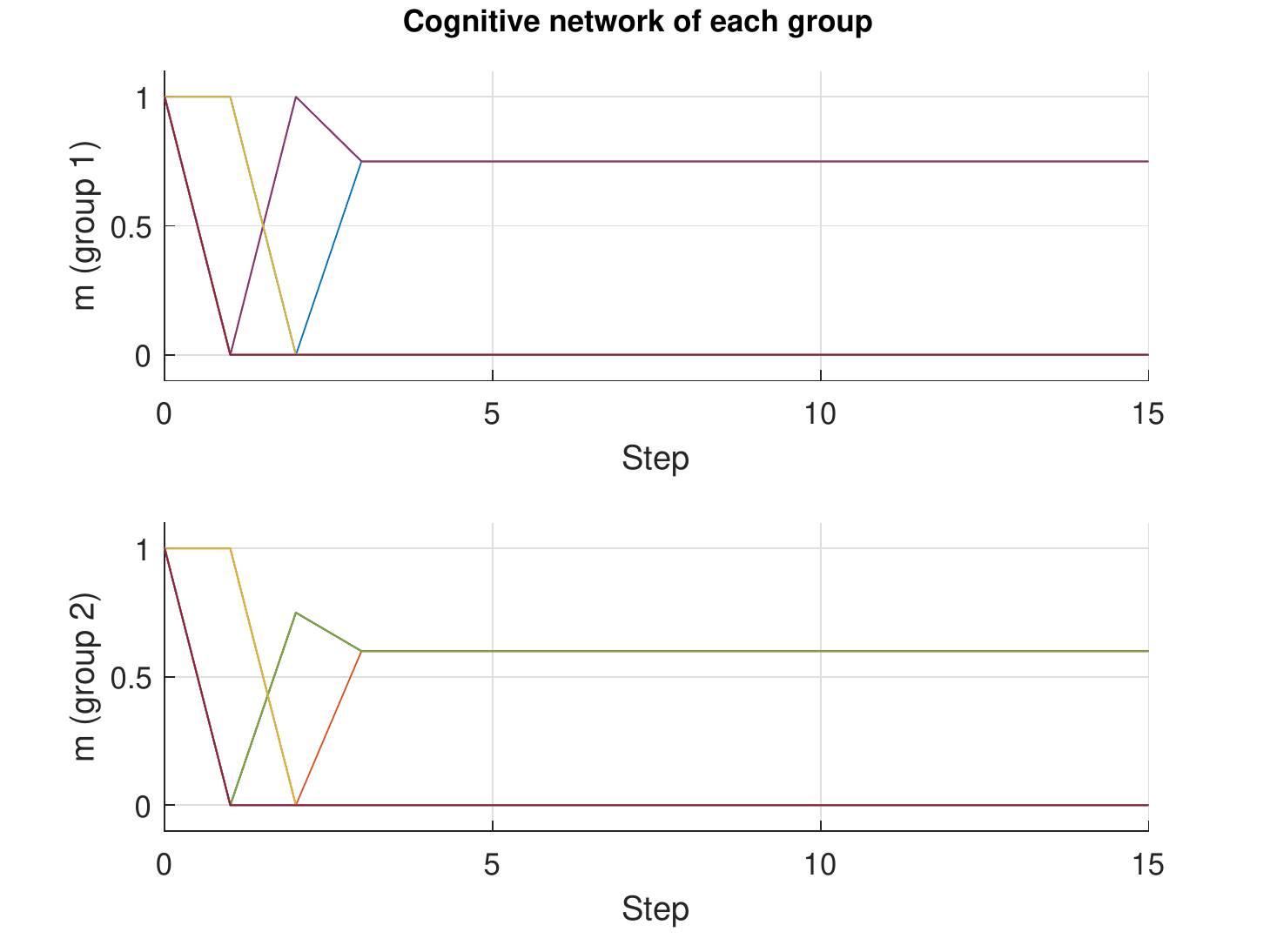}}
            	 \subfigure[cognitive network]{
    \includegraphics[width=0.48\columnwidth]{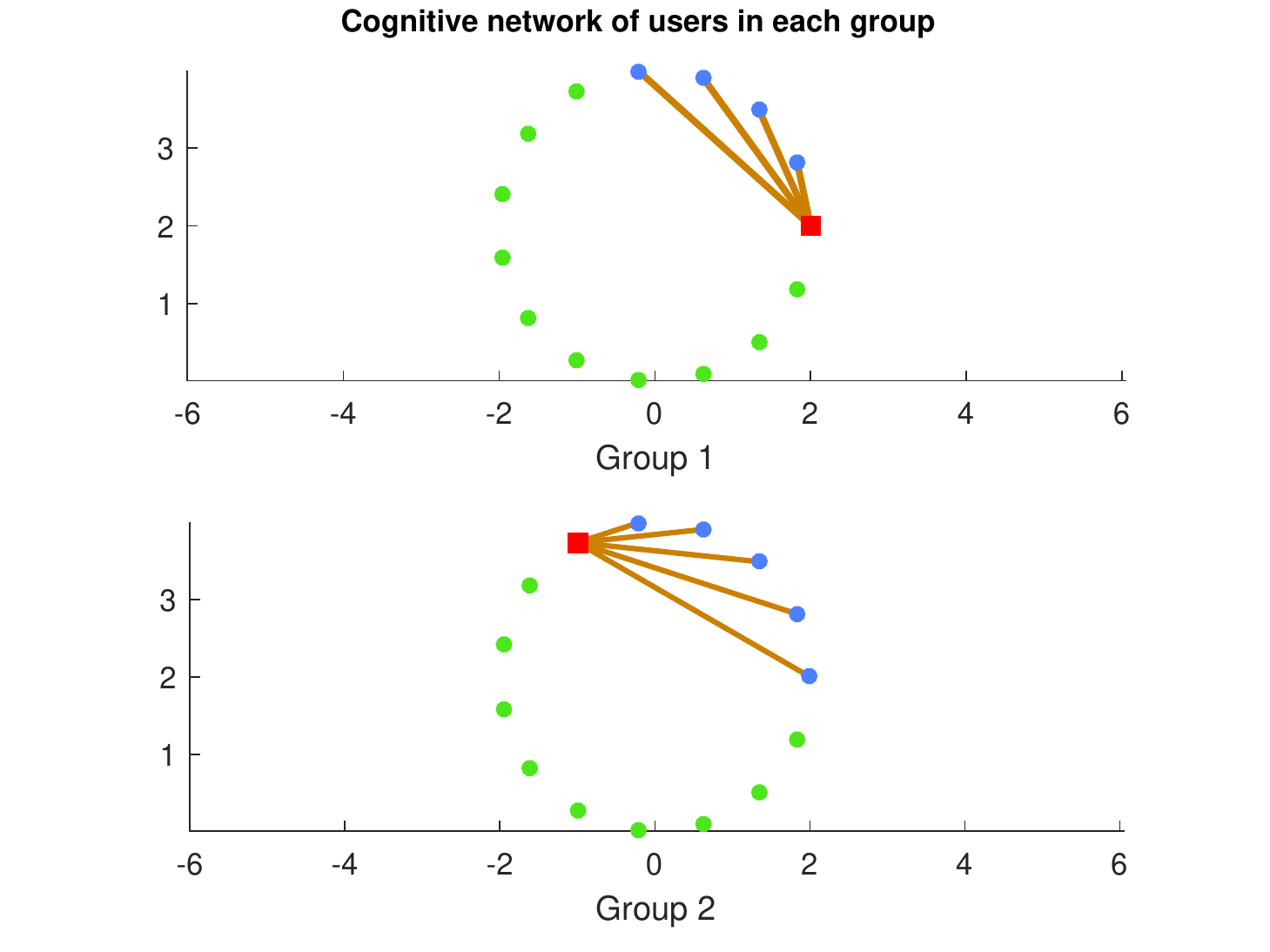}}
  \caption[]{(a) shows the bounded rational strategy of players, indicating that players in $G1$ have a lower cost. (b) depicts the RBP which corroborates that the security risk of users increases under the bounded rational model comparing with the fully rational one. (c) and (d) illustrate the formed sparse cognitive networks. In (d), blue and green dots are agents in $G1$ and $G2$, respectively, and the red ones are representatives in each group. In the network, all agents only allocate cognition resource to smart homes in $G1$ at GNE, leading to the emergence of partisanship.}
  \label{case2}
\end{figure}

\begin{figure}[t]
  \centering
  \subfigure[bounded rational strategy]{
    \includegraphics[width=0.48\columnwidth]{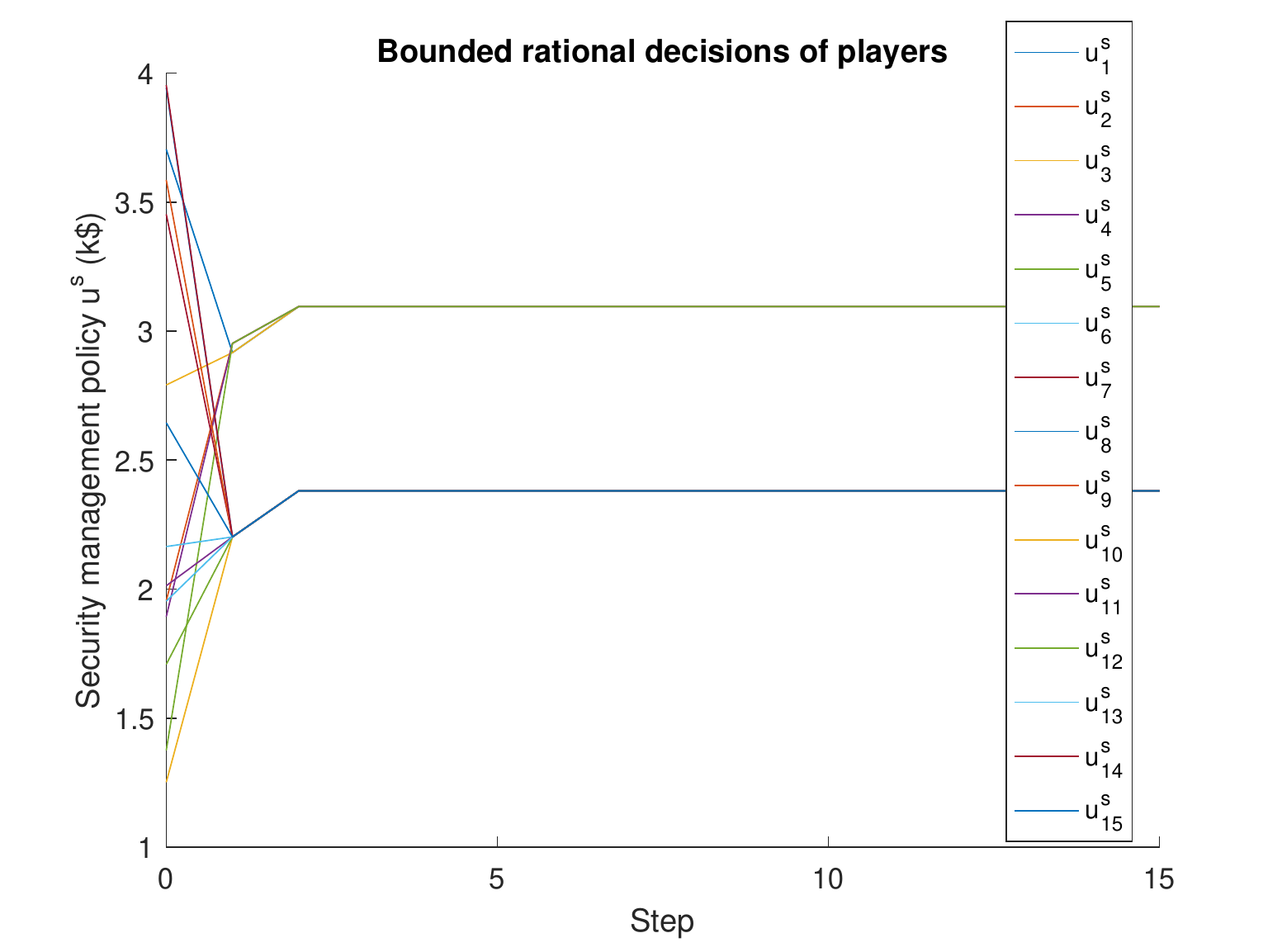}}
    	 \subfigure[RBP]{
    \includegraphics[width=0.48\columnwidth]{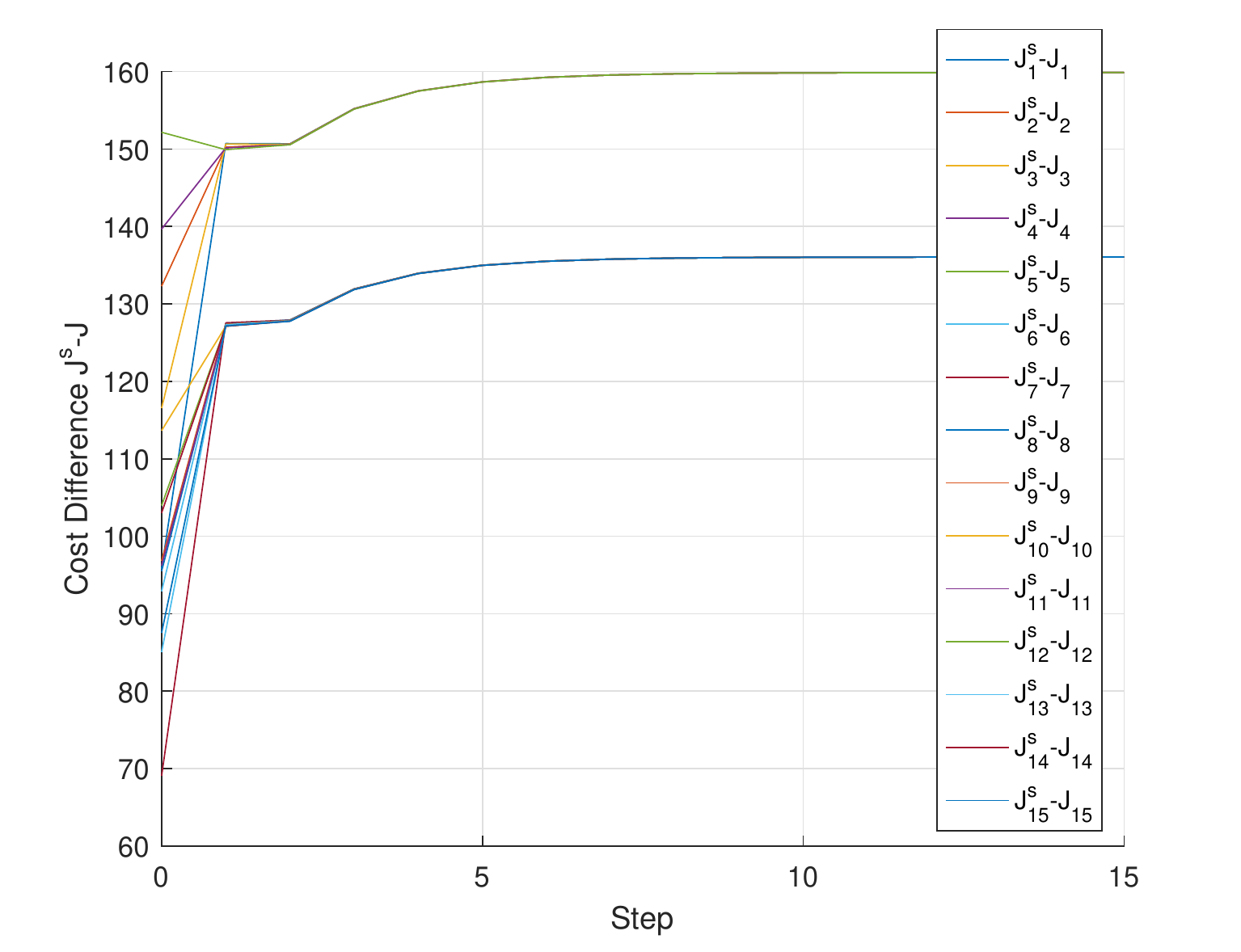}}
        	 \subfigure[cognition vector $m$]{
    \includegraphics[width=0.48\columnwidth]{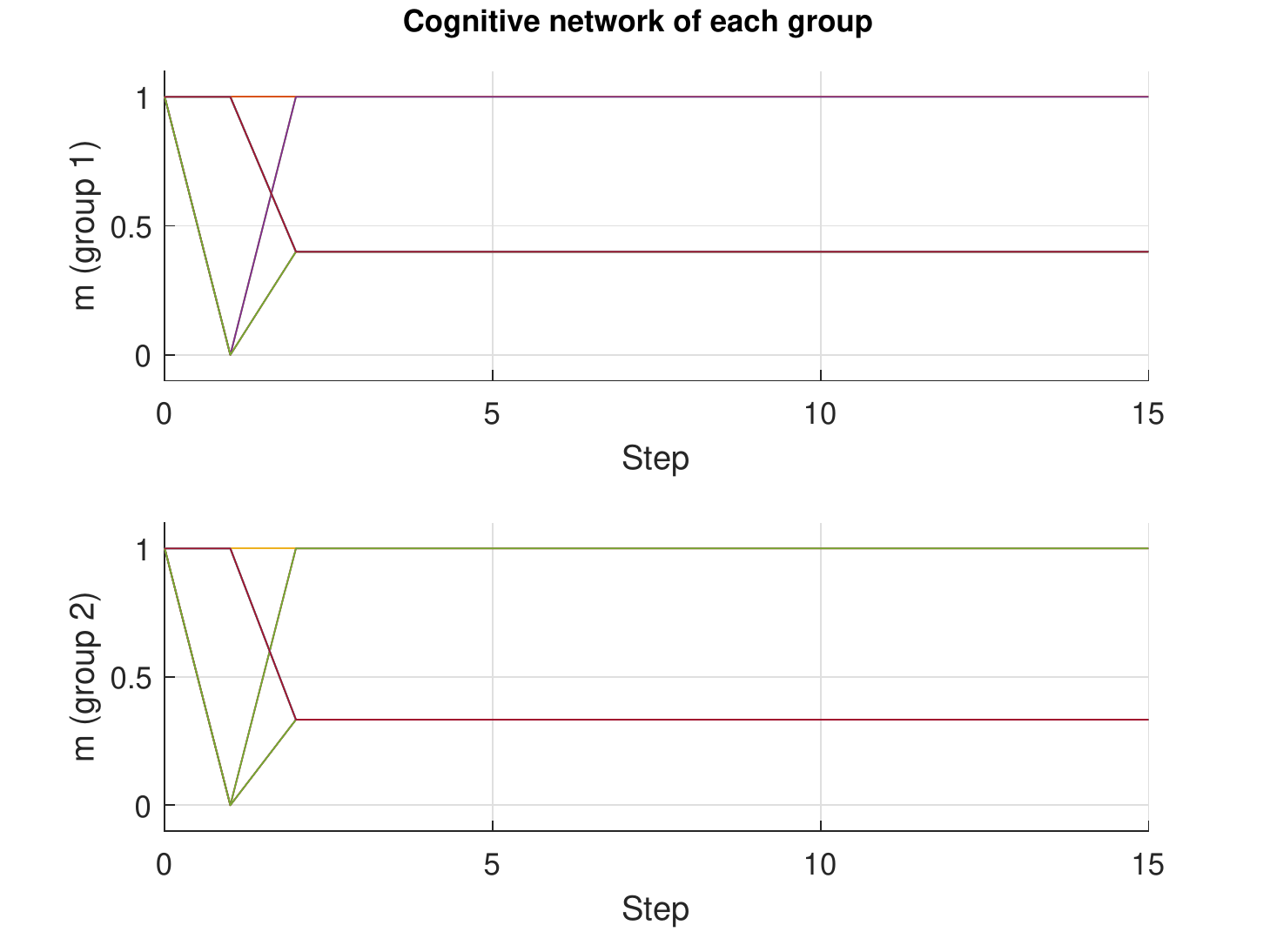}}
        	 \subfigure[cognitive network]{
    \includegraphics[width=0.48\columnwidth]{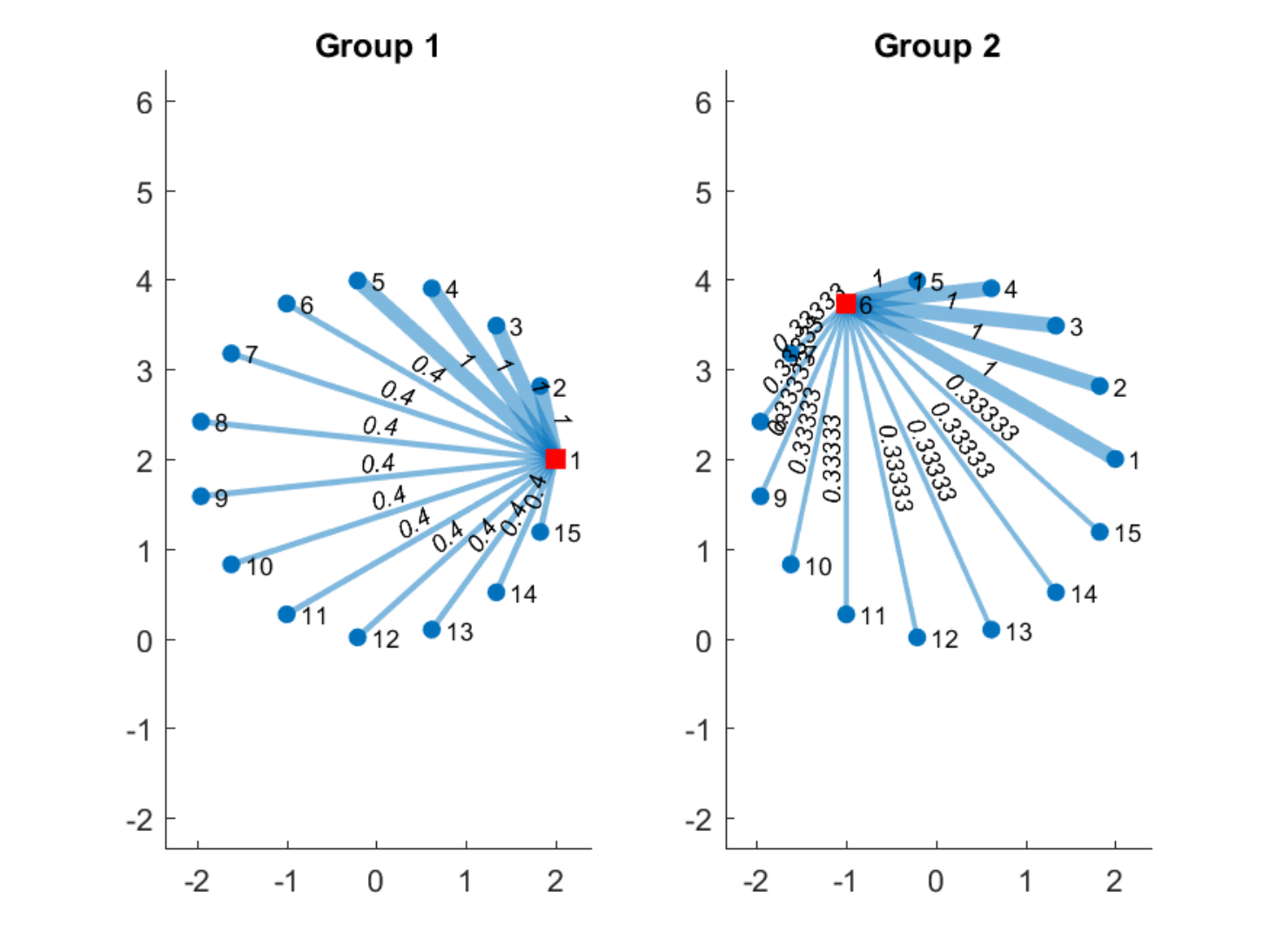}}
  \caption[]{(a) shows the bounded rational decisions, and (b) presents the RBP which is positive. (c) and (d) illustrate the formed cognitive networks. This case study indicates that players in $G1$ are more critical that those in $G2$ in the cognitive networks. In addition, cognition resource is further allocated to the users in $G2$ which reveals the phenomenon of filling the inattention.  }
  \label{case3}
\end{figure}

\begin{figure}[!t]
  \centering
  \subfigure[rational strategy]{
    \includegraphics[width=0.48\columnwidth]{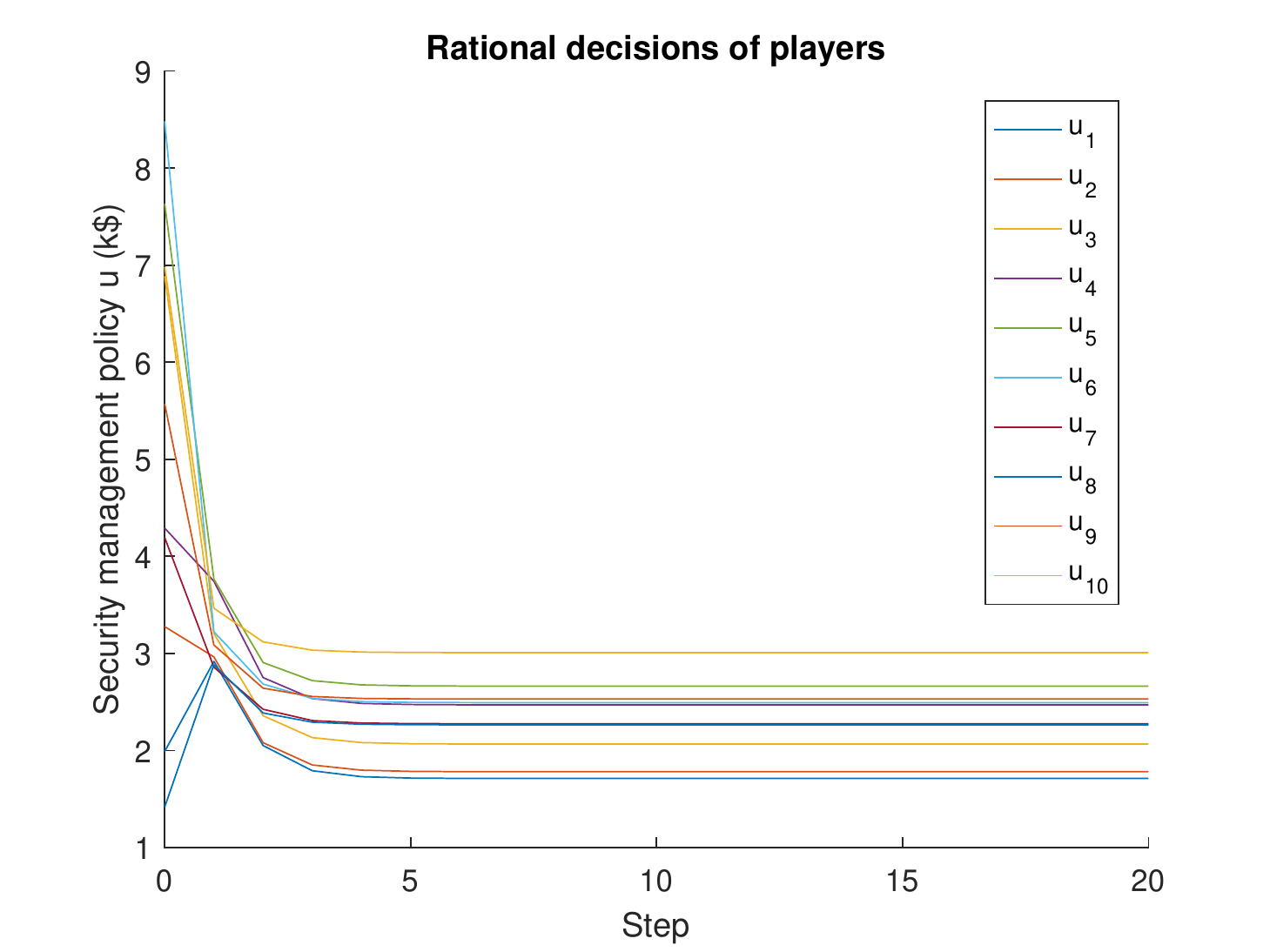}}
    	 \subfigure[cost under rational strategy]{
    \includegraphics[width=0.47\columnwidth]{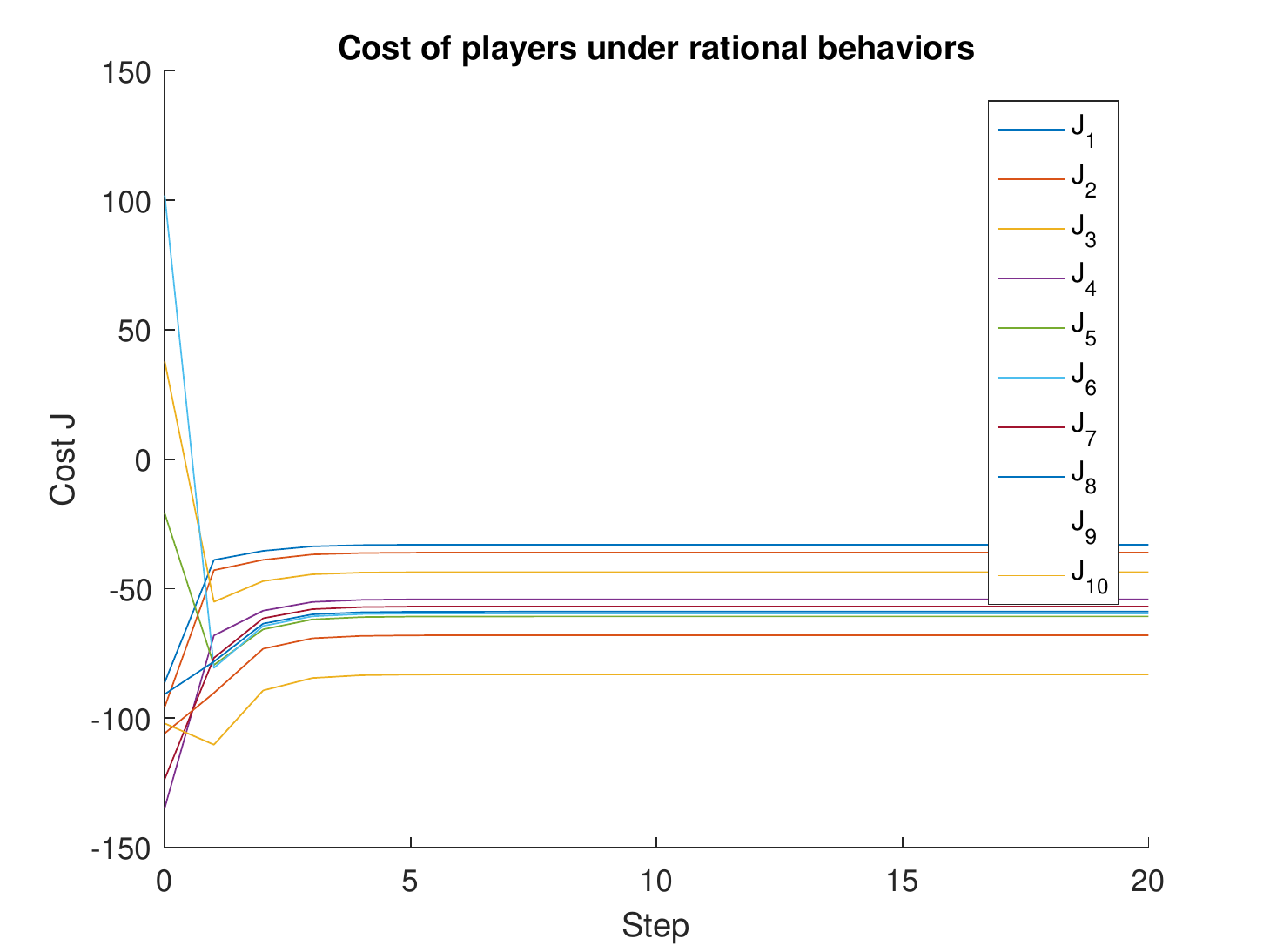}}        	 
    \subfigure[bounded rational strategy]{
    \includegraphics[width=0.47\columnwidth]{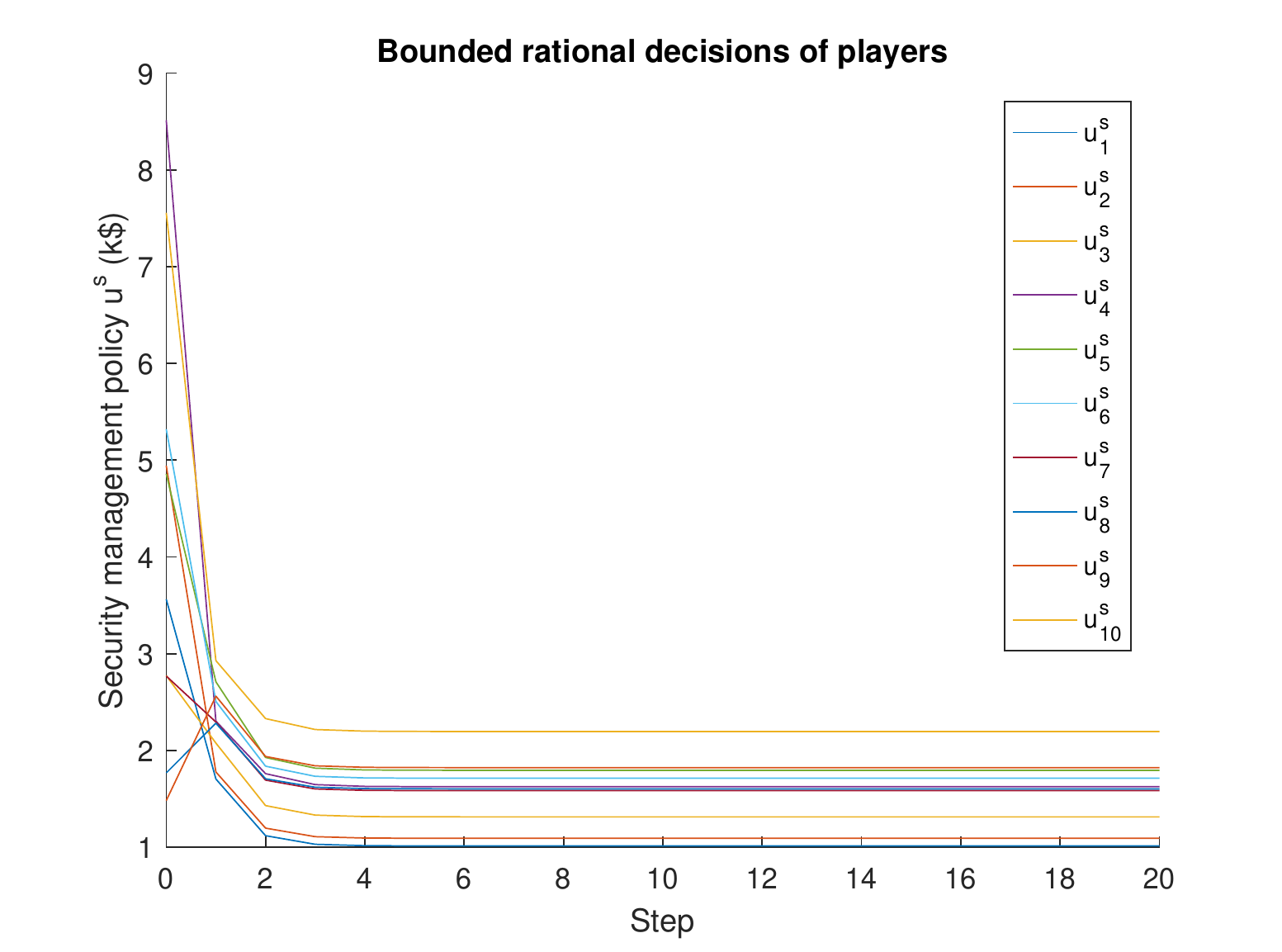}}
            	 \subfigure[RBP]{
    \includegraphics[width=0.47\columnwidth]{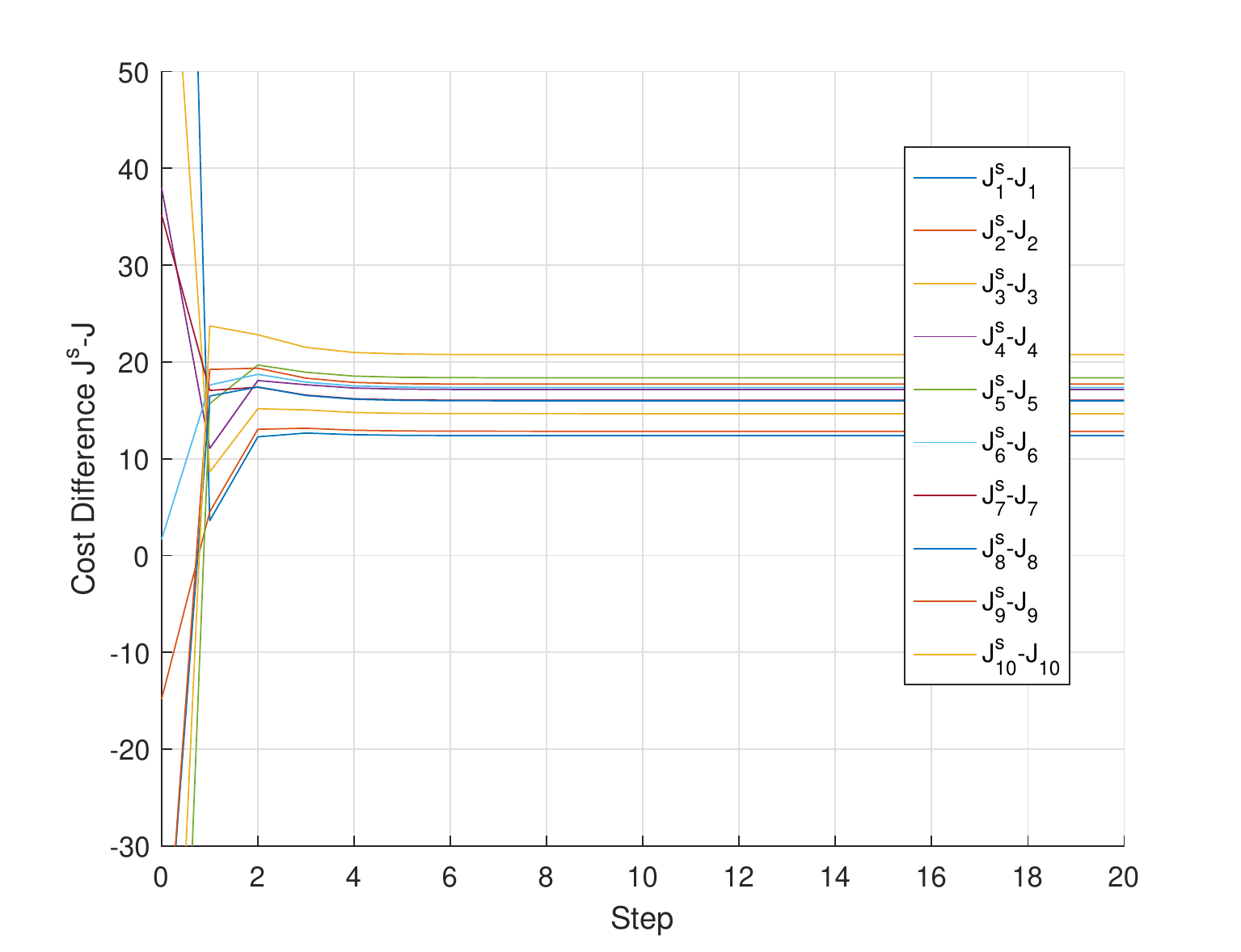}}
            	 \subfigure[cognitive network]{
    \includegraphics[width=0.59\columnwidth]{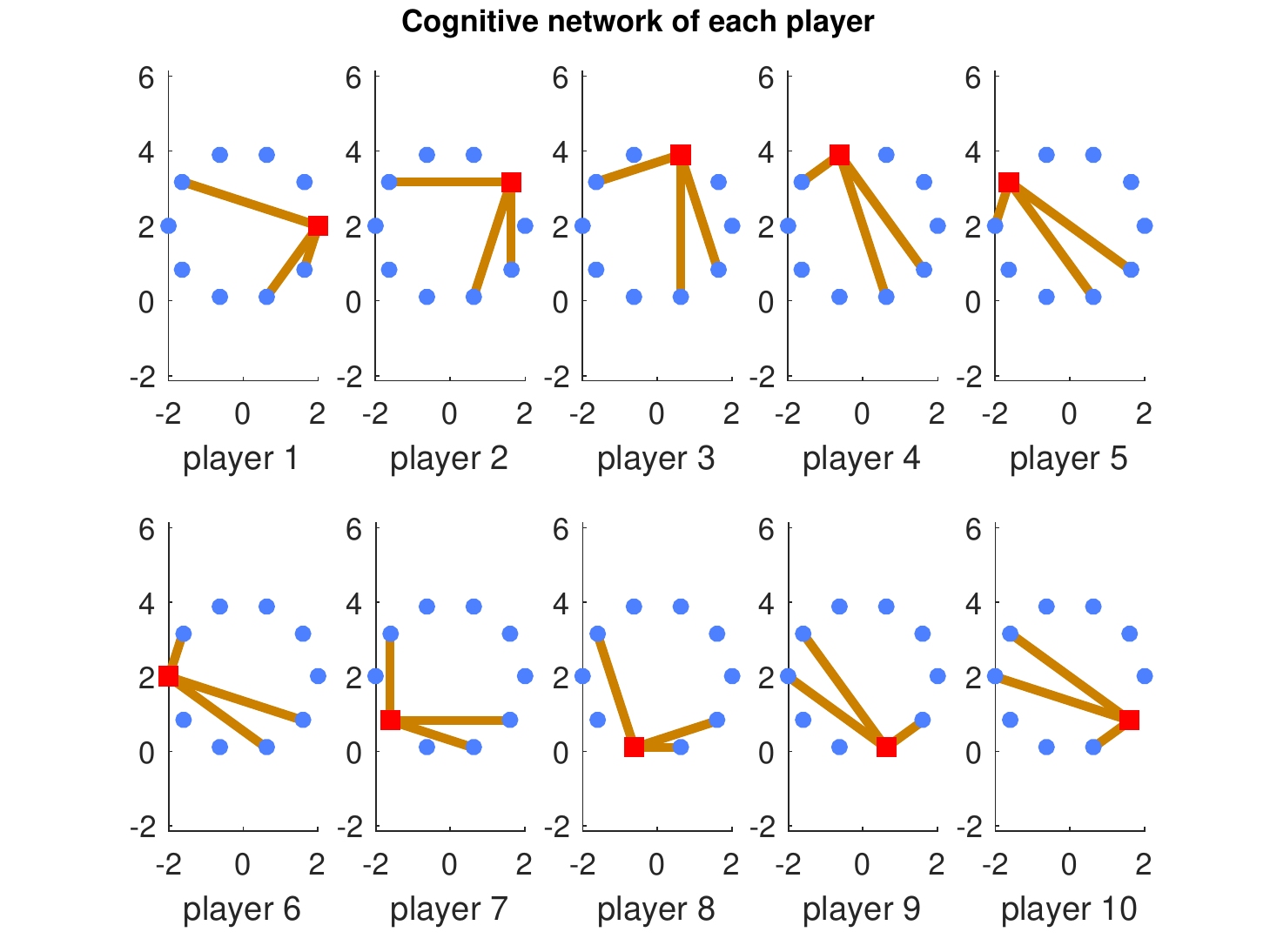}\label{case4_unrational_net}}
  \caption[]{In this heterogeneous case with 10 users, the formed cognitive network shown in (e) is sparse for each smart home. Note that the red rectangular in each subplot of (e) denotes the user that forms his cognitive network with the lines standing for links. Under the bounded rational model, the algorithm can successfully detect the critical agents (attraction of the mighty) in the IoT network which are 5th, 9th and 10th users in this case.}
  \label{case4}
\end{figure}

\subsection{Attraction of the Mighty}
The critical agents in the IoT-enabled smart community are those households whose security management policies will be taken into account by the other agents during their decision makings. Specifically, the nodes who often appear in the cognitive networks of other nodes can be regarded as critical agents. In the following case study, we aim to identify the critical agents in a smart community with $N=10$ households using Algorithm \ref{algorithm2}. To model the heterogeneity of smart homes, we choose $R_{jk}^i = 3\sin(i)+20 \ \mathrm{unit}/\mathrm{k}\$^2$ for $j=k=i$. Otherwise, $R_{jk}^i = 1 \ \mathrm{unit}/\mathrm{k}\$^2$, $\forall i$; $r_i=15+2i \ \mathrm{unit}/\mathrm{k}\$$ for $i\in\mathcal{N}$; and other parameters are the same as those in Section \ref{homo_case}. The results are shown in Fig. \ref{case4}. Specifically, Fig. \ref{case4_unrational_net} shows the established cognitive network of each player. For example, during the cognitive network formation, player 1 chooses to observe the strategies of players 5, 9, and 10 in the network, and player 5's cognitive network includes players 6, 9, and 10. Furthermore, agents 5, 9 and 10 present in all agents' cognitive networks, and hence they constitute a critical community in this smart home network. In addition, agent 6 also plays a critical role in agents 5, 9 and 10's cognitive networks. Therefore, the behavior of agents paying attention to a specific set of households can be described by the attraction of mighty. This case study demonstrates that Algorithm \ref{algorithm2} is able to identify the critical components in the smart communities.

\section{Conclusion}\label{conclusion}
In this paper, we have investigated the security management of users with limited attention over IoT networks through a two-layer framework. The proposed Gestalt Nash equilibrium (GNE) has successfully characterized the bilevel decision makings, including the security management policies and the cognitive network formations of users. Under the security interdependencies, users with a better cognition ability can reduce their cyber risks by making mature decisions. Furthermore, the designed proximal-based algorithm for the computation of GNE has revealed some phenomena that match well with the real-life observations, including the emergence of partisanship and attraction of the mighty. The future work would be extending the framework to incorporate hidden information of unperceived cyber risks of IoT users and design mechanisms to mitigate security loss. Another interesting research direction is to extend the current model to scenarios when a set of users are not fully strategic in minimizing their own risks and analyze the impact of this class of users' misbehavior on the network security risk.

%
%
%

\appendices

\section{Proof of Lemma \ref{lemma1}}\label{lemma1_apx}
\begin{proof}
Based on \eqref{real_cost}, we can compute the RBP of node $i$ as
\begin{equation*}
\begin{split}
&L_i(m^i,u_{-i}) =
 J^i(BR^i(u_{-i}^{c_i}),u_{-i}) - J^i(BR^i(u_{-i}),u_{-i}) \\
&=  \frac{1}{2}\sum_{\substack{j\neq i\\ j\in\mathcal{N}}} \sum_{\substack{k\neq i\\ k\in\mathcal{N}}}   \frac{m_j^i}{R_{ii}^i} R_{ji}^i  R_{ik}^i m_k^i {u_j} u_k  - \frac{1}{2} \sum_{\substack{k\neq i\\ k\in\mathcal{N}}} \sum_{\substack{j\neq i\\ j\in\mathcal{N}}} \frac{m_j^i}{R_{ii}^i} R_{ji}^i  R_{ik}^i {u_j} u_k  \\
&\ + \frac{1}{2} \sum_{\substack{j\neq i\\ j\in\mathcal{N}}}\sum_{\substack{k\neq i\\ k\in\mathcal{N}}}  \frac{1}{R_{ii}^i} R_{ji}^i  R_{ik}^i {u_j} u_k - \frac{1}{2} \sum_{\substack{k\neq i\\ k\in\mathcal{N}}} \sum_{\substack{j\neq i\\ j\in\mathcal{N}}} \frac{m_j^i}{R_{ii}^i} R_{ji}^i  R_{ik}^i {u_j} u_k.
\end{split}
\end{equation*}
Further, we can rewrite 
$\sum_{j\neq i,j\in\mathcal{N}}\sum_{k\neq i,k\in\mathcal{N}}  \frac{1}{R_{ii}^i} R_{ji}^i  R_{ik}^i {u_j} u_k 
= \sum_{j\neq i,j\in\mathcal{N}}\sum_{k\neq i,k\in\mathcal{N}} m_j^i \frac{1}{R_{ii}^i} R_{ji}^i  R_{ik}^i {u_j} u_k 
+  \sum_{j\neq i,j\in\mathcal{N}}(1-m_j^i) \sum_{k\neq i,k\in\mathcal{N}} (1-m_k^i)  \frac{1}{R_{ii}^i} R_{ji}^i  R_{ik}^i {u_j} u_k
 +  \sum_{j\neq i,j\in\mathcal{N}} (1-m_j^i) \sum_{k\neq i,k\in\mathcal{N}}m_k^i \frac{1}{R_{ii}^i} R_{ji}^i  R_{ik}^i {u_j} u_k.
$
Therefore, we obtain
\begin{align*}
L_i&(m^i,u_{-i}) = \frac{1}{2} \sum_{j\neq i,j\in\mathcal{N}}\sum_{k\neq i,k\in\mathcal{N}}  m_j^i \frac{1}{R_{ii}^i} R_{ji}^i  R_{ik}^i m_k^i {u_j} u_k \\
&+ \frac{1}{2} \sum_{j\neq i,j\in\mathcal{N}} (1-m_j^i) \sum_{k\neq i,k\in\mathcal{N}}m_k^i \frac{1}{R_{ii}^i} R_{ji}^i  R_{ik}^i {u_j} u_k \\
&-\frac{1}{2} \sum_{k\neq i,k\in\mathcal{N}} \sum_{j\neq i,j\in\mathcal{N}} m_j^i \frac{1}{R_{ii}^i} R_{ji}^i  R_{ik}^i {u_j} u_k \\
&+ \frac{1}{2} \sum_{j\neq i,j\in\mathcal{N}}(1-m_j^i) \sum_{k\neq i,k\in\mathcal{N}} (1-m_k^i)  \frac{1}{R_{ii}^i} R_{ji}^i  R_{ik}^i {u_j} u_k\\
&= \frac{1}{2} \sum_{j\neq i,j\in\mathcal{N}} \sum_{k\neq i,k\in\mathcal{N}} (1-m_j^i)(1-m_k^i)  \frac{1}{R_{ii}^i} R_{ji}^i  R_{ik}^i {u_j} u_k.
\end{align*}
\end{proof}

\section{Proof of Theorem \ref{convergence}}\label{thm1}
\begin{proof}
 The main idea of the proof follows \cite{frankel2015splitting} with several differences. Especially the imposed conditions for showing convergence in \cite{frankel2015splitting} are different. In addition, our algorithm contains projections and an auxiliary parameter $v_{k+1}$ during updates. 
First, based on Definition \ref{proximal_operator},
$
v_{k+1} = \mathrm{proj}_C\left(\mathrm{prox}_{\lambda_x f_2^i}(x_k- \lambda_x \nabla f_1^i(x_k))\right)
 = \arg\min_{x\in C}\ \langle \nabla f_1^i(x_k),x-x_k \rangle + \frac{1}{2\lambda_x}\Vert x-x_k \Vert^2+f_2^i(x).
$
Then, $\langle\nabla f_1^i(x_k),v_{k+1}-x_k \rangle +\frac{1}{2\lambda_x}\Vert v_{k+1}-x_k \Vert^2+f_2^i(v_{k+1})\leq f_2^i(x_k)$. Based on the Lipschitz continuous condition of $f_1^i$, we obtain
\begin{equation}\label{Q_bound}
\begin{split}
Q_i&(v_{k+1}) \leq f_2^i(v_{k+1})+f_1^i(x_k)+f_3^i(x_k)+ \langle\nabla f_1^i(x_k),v_{k+1}-x_k \rangle \\
&+ \frac{L_i}{2}\Vert v_{k+1}-x_k \Vert^2\\
&\leq  f_2^i(x_k)- \langle\nabla f_1^i(x_k),v_{k+1}-x_k \rangle - \frac{1}{2\lambda_x}\Vert v_{k+1}-x_k \Vert^2 \\
&+f_1^i(x_k)+f_3^i(x_k)+ \langle\nabla f_1^i(x_k),v_{k+1}-x_k \rangle + \frac{L_i}{2}\Vert v_{k+1}-x_k \Vert^2\\
& = Q(x_k) - \left( \frac{1}{2\lambda_x} - \frac{L_i}{2}\right) \Vert v_{k+1}-x_k \Vert^2 .
\end{split}
\end{equation}
When $Q_i(z_{k+1})\leq Q_i(v_{k+1})$, 
$
x_{k+1} = z_{k+1},\ Q_i(x_{k+1}) = Q_i(z_{k+1})\leq Q_i(v_{k+1}),
$
and when $Q_i(z_{k+1})> Q_i(v_{k+1})$, 
$
x_{k+1} = v_{k+1},\ Q_i(x_{k+1}) = Q_i(z_{k+1}).
$
Hence, 
\begin{equation}\label{Q_inequality}
Q_i(x_{k+1})\leq Q_i(v_{k+1})\leq Q_i(x_k).
\end{equation}
Based on \eqref{Q_bound} and \eqref{Q_inequality}, 
\begin{equation}
Q_i(v_{k+1})\leq Q_i(v_k)  - \left( \frac{1}{2\lambda_x} - \frac{L_i}{2}\right) \Vert v_{k+1}-x_k \Vert^2 .
\end{equation}
In addition,
\begin{equation}
\mathrm{dist(0,\partial Q_i(v_{k+1}))}\leq \left( \frac{1}{\lambda_x} + L_i\right) \Vert v_{k+1}-x_k \Vert.
\end{equation}
Furthermore, $\{x_k\}$ and $\{v_k\}$ have the same accumulation points. Let $\Psi$ be the set containing all the accumulation points of $\{x_k\}$. Note that $Q_i$ admits the same value $Q_i^*$ at all accumulation points in $\Psi$ due to the non-increasing $Q_i(v_k)$. Then, $Q_i(v_k)\geq Q_i^*$ and $Q_i(v_k)\rightarrow Q_i^*$. If there exists an $n$ such that $Q_i(v_{n})=Q_i^*$, the algorithm converges. If $Q_i(v_k)\geq Q_i^*$, $\forall k$, then there exists a $\tilde{k}_1$ such that $Q_i(v_k)< Q_i^*+\eta$ for $k>\tilde{k}_1$. Since $\mathrm{dist}(v_k,\Psi)\rightarrow 0$, there exists a $\tilde{k}_2$ such that $\mathrm{dist}(v_k,\Psi)<\epsilon$ for $k>\tilde{k}_2$. Thus, when $k>k_0=\max\{\tilde{k}_1,\tilde{k}_2\}$, $v_k\in \{v,\mathrm{dist}(v_k,\Psi)<\epsilon\}\cap \{Q_i^*<Q_i(v)<Q_i^*+\eta\}$. 
Based on the KL property in Definition \ref{KL_definition}, there exists a concave function $\phi$ such that 
\begin{equation}
\phi'(Q_i(v_k)-Q_i^*)\mathrm{dist}(0,\partial Q_i(v_k))\geq 1.
\end{equation}
Define $r_k := Q_i(v_k)-Q_i^*$, and we further assume that $r_k>0,\ \forall k$. Otherwise, the algorithm converges in finite steps by definition.
Then, $\forall k>k_0$,
\begin{equation}\label{d_1}
\begin{split}
1&\leq \phi'(Q_i(v_k)-Q_i^*)\mathrm{dist}(0,\partial Q_i(v_k))\\
&\leq \left(\phi'(r_k) \left( \frac{1}{\lambda_x} + L_i\right) \Vert v_{k}-x_{k-1} \Vert  \right)^2\\
&\leq (\phi'(r_k))^2 \left( \frac{1}{\lambda_x} + L_i\right)^2  \frac{Q_i(v_{k-1})-Q_i(v_k)}{\frac{1}{2\lambda_x}-\frac{L_i}{2}}\\
&=d_1  (\phi'(r_k))^2 (r_{k-1}-r_k),
\end{split}
\end{equation}
where $d_1 = 2\alpha(\frac{1}{\lambda_x}+L)^2/(1-2\alpha)$. Besides, $\phi$ admits the form $\phi(u) = \frac{\kappa}{\theta}u^\theta$. Then, \eqref{d_1} can be rewritten as 
\begin{equation}
1\leq d_1\kappa^2r_k^{2(\theta-1)}(r_{k-1}-r_k).
\end{equation}
Lemma \ref{KL_f} indicates that $0<\theta\leq \frac{1}{2}$, then, we have $-1\leq \theta-1<-\frac{1}{2}$ and $-1<2\theta-1<0$. When $r_{k-1}>r_k$, we obtain $r_{k-1}^{2(\theta-1)}<r_{k}^{2(\theta-1)}$ and $r_0^{2\theta-1}<r_1^{2\theta-1}<...<r_k^{2\theta-1}$. In addition, define $\zeta(u)=\frac{\kappa}{1-2\theta}u^{2\theta-1}$, and then $\zeta'(u) = -\kappa u^{2\theta-2}$. When $r_{k}^{2(\theta-1)}\leq 2r_{k-1}^{2(\theta-1)}$, then $\forall k>k_0$,
\begin{align*}
\zeta(r_k) - \zeta(r_{k-1}) &= \kappa\int_{r_k}^{r_k}u^{2(\theta-1)}du\geq \kappa r_{k-1}^{2(\theta-1)} (r_{k-1}-r_k)\\
&\geq \frac{1}{2}\kappa r_{k-1}^{2(\theta-1)} (r_{k-1}-r_k) \geq \frac{1}{2\kappa d_1}.
\end{align*}
When $r_{k}^{2(\theta-1)}> 2r_{k-1}^{2(\theta-1)}$, then $r_{k}^{2\theta-1}>2^{\frac{2\theta-1}{2(\theta-1)}} r_{k-1}^{2\theta-1}$, and 
\begin{align*}
&\zeta(r_k) - \zeta(r_{k-1}) =\frac{\kappa}{1-2\theta}(r_{k}^{2\theta-1} - r_{k-1}^{2\theta-1})\\
&> \frac{\kappa}{1-2\theta} (2^{\frac{2\theta-1}{2(\theta-1)}}-1)r_{k-1}^{2\theta-1}> \frac{\kappa}{1-2\theta} (2^{\frac{2\theta-1}{2(\theta-1)}}-1)r_{0}^{2\theta-1}.
\end{align*}
Let $\sigma = \frac{\kappa}{1-2\theta} (2^{\frac{2\theta-1}{2(\theta-1)}}-1)$ and $d_2=\min \{\frac{1}{2\kappa d_1},\sigma r_{0}^{2\theta-1} \}$, then $\forall k>k_0$, $\zeta(r_k)-\zeta(r_{k-1})\geq d_2$, and $\zeta(r_k)\geq \zeta(r_k)-\zeta(r_{k_0})\geq \sum_{t=k_0+1}^k \zeta(r_t)-\zeta(r_{t-1})\geq (k-k_0)d_2$. Hence, $r_k^{2\theta-1}\geq \frac{d_2}{\kappa}(k-k_0)(1-2\theta)$, leading to
$
r_k \leq \frac{\kappa}{d_2(k-k_0)(1-2\theta)}^{\frac{1}{1-2\theta}}.
$
Therefore, we obtain
$
Q_i(x_k)-Q_i^*\leq Q_i(v_k)-Q_i^*=r_k =  \left( \frac{\kappa}{d_2(k-k_0)(1-2\theta)}\right)^{\frac{1}{1-2\theta}}.
$
\end{proof}

\bibliographystyle{IEEEtran}
\bibliography{IEEEabrv,references}

\end{document}